\documentclass[11pt]{article}

\usepackage[utf8]{inputenc}
\usepackage[margin=1in]{geometry}
\usepackage{paralist}
\usepackage[full,funcfont=roman]{complexity}
\usepackage{xcolor}

\usepackage{amsmath,amssymb}
\usepackage{amsthm}
\usepackage{thmtools,mathtools}

\usepackage[algoruled,boxed,lined]{algorithm2e}
\usepackage[full,funcfont=roman]{complexity}
\usepackage{mathtools}
\usepackage[textsize=tiny]{todonotes}
\usepackage{thm-restate}

\newcommand{\LOCAL}{$\mathsf{LOCAL}$\xspace}
\newcommand{\SLOCAL}{$\mathsf{SLOCAL}$\xspace}

\def\NN{{\mathbb N}}

\newcommand{\eps}{\varepsilon}
\def\logstar{{\log^*}}

\DeclarePairedDelimiter{\ceil}{\lceil}{\rceil}
\DeclarePairedDelimiter{\floor}{\lfloor}{\rfloor}
\DeclarePairedDelimiter{\parenths}{(}{)}
\DeclarePairedDelimiter{\brackets}{[}{]}
\DeclarePairedDelimiter{\curlybrackets}{\{}{\}}

\newtheorem{theorem}{Theorem}[section]
\newtheorem{claim}[theorem]{Claim}
\newtheorem*{theorem*}{Theorem}
\newtheorem{lemma}[theorem]{Lemma}
\newtheorem{corollary}[theorem]{Corollary}

\newtheorem{observation}[theorem]{Observation}

\newtheorem{fact}[theorem]{Fact}

\theoremstyle{definition}
\newtheorem{definition}[theorem]{Definition}

\theoremstyle{remark}
\newtheorem{remark}[theorem]{Remark}

\usepackage{hyperref}

\hypersetup{hypertexnames=false}

\usepackage[nameinlink,capitalize]{cleveref}
\hypersetup{
	colorlinks=true,
	linkcolor=blue,
	filecolor=magenta,      
	urlcolor=cyan,
}

\Crefname{remark}{Remark}{Remarks}
\Crefname{observation}{Observation}{Observations}
\Crefname{property}{Property}{Properties}
\Crefname{definition}{Definition}{Definitions}
\Crefname{fact}{Fact}{Facts}

\hypersetup{
	colorlinks=true,
	linkcolor=blue,
	filecolor=magenta,      
	urlcolor=cyan,
}
\Crefname{remark}{Remark}{Remarks}
\Crefname{observation}{Observation}{Observations}
\Crefname{property}{Property}{Properties}
\Crefname{definition}{Definition}{Definitions}

\newcommand{\myemail}[1]{\,$\cdot$\, {\small #1}}
\newcommand{\myaff}[1]{\,$\cdot$\, {\small #1}\par\smallskip}

\newenvironment{myabstract}
{\list{}{\listparindent 1.5em%
		\itemindent    \listparindent
		\leftmargin    1cm
		\rightmargin   1cm
		\parsep        0pt}%
	\item\relax}
{\endlist}

\newenvironment{mycover}
{\list{}{\listparindent 0pt
		\itemindent    \listparindent
		\leftmargin    1cm
		\rightmargin   0.5cm
		\parsep        0pt}%
	\raggedright
	\item\relax}
{\endlist}

\begin{document}

\begin{mycover}
	{\huge\bfseries\boldmath Distributed Vertex Cover Reconfiguration \par}
	\bigskip
	\bigskip
	\bigskip
	
	\textbf{Keren Censor-Hillel}
	\myemail{ckeren@cs.technion.ac.il}
	\myaff{Technion, Israel}

	\textbf{Yannic Maus}
	\myemail{yannic.maus@ist.tugraz.at}
	\myaff{TU Graz, Austria}
	
	\textbf{Shahar Romem-Peled}
	\myemail{shaharr@campus.technion.ac.il}
	\myaff{Technion, Israel}
	
	\textbf{Tigran Tonoyan}
	\myemail{ttonoyan@gmail.com}
	\myaff{Technion, Israel}
\end{mycover}
\bigskip

\thispagestyle{empty}
\begin{myabstract}

Reconfiguration schedules, i.e., sequences that gradually transform one solution of a problem to another while always maintaining feasibility, have been extensively studied. Most research has dealt with the decision problem of whether a reconfiguration schedule exists, and the complexity of finding one. A prime example is the reconfiguration of vertex covers. 
We initiate the study of \emph{batched vertex cover reconfiguration}, which allows to reconfigure multiple vertices concurrently while requiring that any adversarial reconfiguration order within a \emph{batch} maintains feasibility.  The latter provides robustness, e.g., if the  simultaneous reconfiguration of a batch cannot be guaranteed. The quality of a  schedule is measured by the number of batches until all nodes are reconfigured, and its \emph{cost}, i.e., the maximum size of an intermediate vertex cover.

 To set a baseline for batch reconfiguration, we show that for graphs belonging to one of the classes $\{\mathsf{cycles, trees, forests, chordal, cactus, even\text{-}hole\text{-}free, claw\text{-}free}\}$, there are schedules that use $O(\eps^{-1})$ batches and incur only a $1+\eps$ multiplicative increase in cost over the best sequential schedules. Our main contribution is to compute such batch schedules in  $O(\eps^{-1}\logstar n)$ distributed time, which we  also show to be tight. Further, we show that once we step out of these graph classes we face a very different situation.  There are  graph classes on which no efficient distributed algorithm can obtain the best (or almost best) existing schedule. Moreover, there are classes of bounded degree graphs which do not admit any reconfiguration schedules without incurring a large  multiplicative increase in the cost at all.
\end{myabstract}

\clearpage

\clearpage

\tableofcontents
\clearpage
\setcounter{page}{1}
\section{Introduction}
Consider a huge network of computers connected via communication links, in which each communication link needs to be monitored at all times by at least one of its endpoints (computers). If the network is abstracted as a graph with each node representing a computer and each edge representing a communication link, the set of all monitoring computers is a \emph{vertex cover}. A usual constraint is to have small vertex covers (consisting of few computers, in our case), and the problem of finding minimal vertex cover is one of the classic problems of algorithmic graph theory and has been extensively studied in various graph classes and computational models. In the setting we are interested in, the system might decide at some point in time to switch to another vertex cover, say, to evenly distribute the load of monitoring over the nodes. To ensure correct system performance, it is natural to require each communication link be monitored at all times, even during the switching process, and at the same time to save resources by keeping the vertex cover size small at all times. This naturally leads us to the vertex cover reconfiguration problem.

More generally, reconfiguration problems ask the following type of questions: 
\emph{Given two solutions to a problem, 
is it always possible to gradually move from one solution to the other by changing \textbf{one element} at a time, while always maintaining a feasible solution?}

Reconfiguration problems thus explore reachability in a graph over the solutions, and as such they have been extensively studied for various problems. Notable examples are colorings~\cite{CerecedaHJ08,ItoKD12,BonamyHI0MMSW20}, matchings~\cite{ItoKK0O19}, independent sets and vertex covers~\cite{ KaminskiMM12,LokshtanovM19}. 
 In the vertex cover reconfiguration problem  one needs to find a schedule that moves from
a given first vertex cover to a given second vertex cover by changing the membership of one vertex at a time, and while 
ensuring that each intermediate set is a valid vertex cover (\textbf{feasibility}). 
Traditionally, the emphasis has been on the size of the intermediate solutions---the problem is trivial with no size bound---while the elementary steps consist of adding or removing a single vertex.  

\subparagraph*{Distributed reconfiguration and its motivation.}
In this work, we initiate the study of \emph{distributed vertex cover reconfiguration}. In contrast to the previously discussed setup, 
in huge networks, the states of vertices might have to change concurrently, for various reasons; e.g., to obtain short schedules, or simply if there is no single entity controlling the whole network, and communication between far-apart computers is too costly to agree on a global reconfiguration schedule. 
Motivated by this,
the goal of this work is to exploit parallelism for reconfiguration schedules.  For vertex cover reconfiguration, this leads us to allow changing the membership in the intermediate vertex cover of more than one vertex in each step.   We formally capture this setting by introducing the concept of batch reconfiguration.

In \emph{batch reconfiguration}, one is allowed to change a \emph{batch} of an \textbf{unbounded number of elements}  in a single reconfiguration step, as opposed to the previous \emph{sequential reconfiguration}, 
which changes a single vertex at a time. 
However, such a solution is not practically robust, in the following sense. Suppose that implementing the change for a vertex is not an immediate operation and could rather take a bit of time. Then, changing several vertices 
concurrently
may result in a sequence of changes in these vertices, with an unpredictable order. As a result, although we aim at reconfiguring all vertices at once in one swipe, in reality what could happen is that we get an arbitrary sequence of changes, which can easily violate  feasibility in an adversarial execution of a batch.

In light of the above, in addition to feasibility, we require a \textbf{robustness} condition for batch reconfiguration schedules. The goal of a robustness condition is to guarantee that no matter in which order the elements of a batch are eventually executed, feasibility is never violated. We require that \emph{the set of vertices that are reconfigured within a batch is always an independent set}. This promises that each edge is always covered, also within any internal ordering of a batch.

\subparagraph*{The tradeoff between the number of batches and the solution size.}
When computing a schedule between two covers, typically denoted by $\alpha$ and $\beta$, batching brings the advantage of short schedules, but comes with a proportional overhead in solution sizes. We show via a pigeonhole argument that for some instances of vertex covers $\alpha$ and $\beta$, every reconfiguration schedule with $t$ batches necessarily creates an intermediate solution of size at least $(1+1/t)\cdot \max\{|\alpha|,|\beta|\}$.
We desire to get as close as possible to this optimal length-vs.-size tradeoff. We 
evaluate the quality of a batch reconfiguration schedule for vertex cover by two measures: the length of the schedule, i.e., the number of batches, and the maximum size of an intermediate vertex cover in the schedule, which, by the robustness condition, includes all possible intermediate vertex covers that can occur in any internal ordering of each batch.
In view of the natural barrier of $\max\{|\alpha|,|\beta|\}$ on the worst-case intermediate vertex cover size, as well as the necessary overhead given by batch schedules, we say a reconfiguration schedule is an \emph{$(\eta,c)$-approximation} if the size of the worst-case intermediate cover is at most $\eta\cdot \max\{|\alpha|,|\beta|\} + c$ (see \cref{sec:defs} for formal definitions). Here, $\eta$ is between $1$ and $2$, and is normally related to the number of batches, while $c$ is often a constant (e.g., the trivial $2$-batch schedule that adds all nodes in $\beta\setminus \alpha$, then removes all nodes in $\alpha\setminus \beta$, is a $(2,0)$-approximation for a schedule from $\alpha$ to $\beta$).

\subsection{Our Contribution}
As we initiate the study of distributed algorithms for vertex cover reconfiguration, our first contribution is an examination of the baseline for our distributed algorithms. We provide several (tight) existential results on batch schedules on various graph classes. Then we devise  distributed algorithms that nearly match our existential results and efficiently compute schedules of almost the same quality. Further, we show that once we step out of these graph classes we face a very different situation.  There are  graph classes on which no efficient distributed algorithm can obtain the best (or almost best) existing schedule. Moreover, there are classes of bounded degree graphs which do not admit any reconfiguration schedules without incurring a large  multiplicative increase in the cost at all.
 We next discuss our existential (centralized) results and then discuss our distributed results.

\subparagraph*{Existential results on batch reconfiguration schedules.}
Our first technical contribution is a black-box compression scheme that mechanically transforms a sequential schedule into a batched one of desired length, with a proportional and unavoidable overhead. Thus, we derive batch schedules from known sequential schedules. Our result holds for \emph{monotone}   schedules (that never touch a vertex twice),
which is aligned with prior work.

\begin{restatable}[Schedule Compression]{theorem}{schedCompress} \label{thm::sched_compres}Let $G=(V,E)$ be a graph with two vertex covers $\alpha,\beta$, and let $\mathcal{S}$ be a monotone  sequential schedule from $\alpha$ to $\beta$ that is an $(\eta,c)$-approximation, for a real $\eta\ge 1$ and an integer $c\ge 0$. For every $\eps\in (0,1)$, $\mathcal{S}$ can be transformed into a monotone $(2\lceil 1/\eps\rceil)$-batch schedule $\mathcal{S}'$ that is an $(\eta+\eps, c+1)$-approximation. 
\end{restatable}

Based on known separator theorems (e.g.,~\cite{LiptonT80}), we observe that a number of graph classes possessing small separators admit $(1,O(\sqrt{n}))$-approximation schedules, where $n$ is the number of vertices. 
It is also known that cactus graphs have $(1,2)$-approximation schedules~\cite{ItoNZ16} (see \Cref{thm:sequential}). We provide a simple and unified proof of these results, and combined   with our batching scheme to get the following theorem.

\begin{restatable}{theorem}{variousClassesBatch}
\label{thm:variousClassesBatch}
Let $\eps>0$. Let $G=(V,E)$ be a graph with a pair $\alpha,\beta$ of vertex covers. If $G$ belongs to one of the  graph classes \emph{\{\textsf{cactus, chordal, even-hole-free, claw-free}\}},  there is a monotone $(2\ceil{1/\eps})$-batch $(1+\eps,3)$-approximation schedule from $\alpha$ to $\beta$.
If $G$ is planar,  there is a monotone $(2\ceil{2/\eps})$-batch $(1+\eps,O(\eps^{-1}))$-approximation schedule from $\alpha$ to $\beta$. 
\end{restatable}

We show in \cref{sec:centralized}  that the tradeoff between the number of batches and the approximation overhead in \Cref{thm::sched_compres} is (almost)  best possible. This implies that the batch schedules from \Cref{thm:variousClassesBatch} are nearly optimal in terms of the length vs. approximation tradeoff.

A natural guess would be that the $(1+\eps)$-type approximations \'a la \Cref{thm:variousClassesBatch} extend to bounded degree or bounded arboricity graphs\footnote{The \emph{arboricity} of a graph is the minimum number of forests that are needed to cover its edge set.} (the latter contains planar graphs). 
In \cref{app:sequential}, we show that that is not the case: for any $d\ge 4$, there are infinitely many $d$-regular Ramanujan graphs $G$ (\cite{MarcusSS13}) with two vertex covers $\alpha,\beta$, for which no schedule can be better than a $(2-O(1/\sqrt{d}))$-approximation. The proof exploits the \emph{expansion} properties of such graphs. On the positive side, we give a $(2-\Omega(1/\lambda))$-approximation for any graph of arboricity $\lambda$, by a simple monotone  \emph{greedy schedule} that reconfigures $\alpha$-nodes in an increasing degree order, and  
a similar approximation can even be achieved with a \emph{4-batch} schedule. %These results (except for the batched variant) are summarized in the following theorem.

\begin{restatable}{theorem}{thmarboricity}\label{thm:arboricitybounds}
For every $d\ge 4$, there is an infinite class of $d$-regular graphs $G_i$, $i\ge 1$, with vertex covers $\alpha_i,\beta_i$, that do not admit $\left(2-4/(\sqrt{d+1}+1),0\right)$-approximation schedules. Any graph $G$ of arboricity $\lambda$ with vertex covers $\alpha,\beta$ has a $(2-1/(2\lambda),1)$-approximation schedule.
\end{restatable}

Despite motivation through distributed reconfiguration we believe that batch reconfiguration is a meaningful concept in itself and adds to the notions of parallelism that can be introduced into solving a problem. Thus we believe that our existential results are also of independent interest. 

\subparagraph*{Distributed computation of reconfiguration schedules.} A natural setting where batch reconfiguration may appear is when nodes themselves, as distributed autonomous computers, execute a reconfiguration schedule by exchanging information with others. This immediately raises the question of \emph{computing} schedules in such distributed settings.  

We focus on the \LOCAL model of distributed computing, where nodes synchronously send messages to their neighbors in the underlying network graph (which also serves as the problem instance), and give efficient algorithms for computing  batch schedules.

At the core of most of our algorithmic results is a graph decomposition, that we call a \emph{small separator decomposition} and which might be of independent interest.  Roughly speaking (and hiding many technical details), it is a decomposition into small diameter clusters and a  small  separator set $S$ (e.g., an $\eps$ fraction of $|\alpha|+|\beta|$, for a small $\eps>0$), such that there are no \emph{inter-cluster edges}. At a very high level, we show that if there exists a good approximation schedule for each cluster, we can compute a batch schedule for the whole graph with only a small approximation overhead and in a number of rounds that is linear in the maximum cluster diameter.

\begin{theorem*}[Informal version of \cref{gen::lma::decomp_to_sched_computation}]
Let $G=(V,E)$ be a graph with two vertex covers $\alpha$ and $\beta$, and a small separator decomposition with
a separator set $S$ and
clusters $(C_i)_{i=1}^k$ of diameter $d$. If each subgraph $G[C_i]$ admits a monotone $\ell$-batch  $(\eta,c)$-approximation schedule,
then a $(2\ell+2)$-batch 
$(\eta,|S|+k\cdot c)$-approximation schedule for $\alpha$ and $\beta$
can be computed in $O(d)$ rounds in the \LOCAL model.
\end{theorem*}

As can be seen in the theorem above, it is crucial to have a graph decomposition both with a small size separator, as well as few clusters.  Combining known network decompositions \cite{awerbuch89,linial93,Vaclav_2019} and tools from distributed combinatorial optimization \cite{SLOCAL17} together with our new machinery, we get a $\poly\log(n)$-round algorithm to compute a small separator decomposition with logarithmic-diameter clusters in  general graphs. While the resulting decomposition has a small separator set, we use  additional postprocessing to also reduce the number of clusters (while keeping the bounds on  cluster diameters and the separator size). This allows us to distributively implement batch-scheduling algorithms known for any \emph{hereditary} graph class. For instance, in $\poly\log (n)$ rounds, we can construct an $O(1/\eps)$-batch $(1+\eps,O(\eps^{-1}))$-approximation schedules for planar graphs. See \cref{sec:SLOCAL} and
\Cref{thm:distrPlanar} for details.

For other special graph classes, we obtain much stronger results, that is,  we can compute small separator decompositions \emph{super-fast} (i.e., in $O(\log^* n)$-rounds). This holds in particular for \emph{cactus} graphs. More concretely (see \Cref{lem:cactus_core_decomp}), there is a super-fast distributed algorithm that for any connected graph $G$ and vertex covers $\alpha,\beta$ such that $\alpha\oplus\beta$ induces a cactus graph, computes a small
separator decomposition with few clusters. 
The property ``$G[\alpha\oplus\beta]$ is a cactus graph'' holds for all graphs $G$ in the following  graph classes
$\{\mathsf{cycles, trees, forests, chordal, cactus, even\text{-}hole\text{-}free, claw\text{-}free}\}$.

For such graphs $G$, 
we have $(1,2)$-approximation schedules \cite{ItoNZ16} (see also \Cref{thm:sequential}). 
Together with \Cref{gen::lma::decomp_to_sched_computation,lem:cactus_core_decomp}, this implies the following theorem.

\begin{restatable}[Cactus Reconfiguration]{theorem}{thmCactusCoreRecon}\label{thm:cactus_core_graphs}For each $\eps>0$, there is a  \LOCAL algorithm that for any connected $n$-node graph $G$ and vertex covers $\alpha$ and $\beta$ such that $G[\alpha\oplus\beta]$ is a cactus graph, computes an $O(1/\eps)$-batch $(1+\eps,3)$-approximation schedule in $O(\logstar(n)/\eps)$ rounds.
\end{restatable}

We also show that \emph{the runtime of \Cref{thm:cactus_core_graphs}  is asymptotically tight}. Concretely, we use an argument based on Ramsey Theory to show that for every $\eps\in (0,1)$, there is no algorithm with runtime $c\cdot \logstar(n)$, with a small enough constant $c>0$, that constructs $(1/\eps)$-batch  $(2-\eps)$-approximation schedules on the class of \emph{cycle graphs}  (see \Cref{thm:logstarLB}).

So far, we have seen that we can obtain $\approx (1+\eps)$-approximation  for planar graphs, in $\poly\log n$ rounds, and for cactus graphs, in $O(\logstar n )$ rounds. The next theorem provides  super-fast algorithms for batch scheduling for general graphs, at the cost of an increased schedule length and approximation (which, by \cref{thm:arboricitybounds},  cannot be improved by much).

\begin{restatable}{theorem}{arbOvershooting}
\label{thm:arbOvershooting}
Let  $G$ be a graph with vertex covers $\alpha$ and $\beta$, such that $G[\alpha\oplus\beta]$ has arboricity at most $\lambda$, with $\lambda$  known to all nodes. For every $\eps>0$, there exists a  $O\parenths*{\log^* (n)/\eps}$-round \LOCAL algorithm that computes an $O\left(\lambda/\eps^2\right)$-batch $\left(2-1/(2\lambda)+\varepsilon ,1\right)$-approximation  schedule from $\alpha$ to $\beta$. \end{restatable}

The algorithm from \Cref{thm:arbOvershooting} is a non-trivial adaptation of the greedy algorithm from \Cref{thm:arboricitybounds} to the distributed setting. We briefly discuss the underlying ideas in Section~\ref{sec:mainarboverview}. 

\cref{gen::lma::decomp_to_sched_computation} suggests that if every cluster in the decomposition admits a good reconfiguration schedule, one can compute a good schedule for the whole graph efficiently.
We show that this condition is necessary, in the following strict sense: there are graphs for which there exist good schedules, and they can be quickly decomposed, but just because some clusters do not have good schedules, no \emph{sub-linear} distributed algorithm can compute a good schedule.

\begin{theorem*}[Informal version of \Cref{thm::unwighted_gen_lb}]
	\label{thm:informaldistrLB}
For each $\eps\in (0,1)$, there is an infinite family $\mathcal{G}$ of graphs such that every graph in $\mathcal{G}$ admits a $(1+\eps)$-approximation schedule, but no $o(n)$-round deterministic distributed algorithm can compute a $(2-\eps)$-approximation schedule.
\end{theorem*}

\subsection{Related Work}
Reconfiguration problems have long been studied under various guises (e.g., in the form of puzzles \cite{HearnD05}), but a more systematic study has appeared rather recently, in \cite{ITO2011}. The high level picture of the area is that for NP-complete source problems, the reconfiguration variants are usually \PSPACE-complete \cite{ITO2011,BonsmaC09,HearnD05}, although there are exceptions to this trend \cite{JohnsonIPEC14}. In this paper, we are only concerned with vertex cover/independent set reconfiguration; for a wider view on the subject, we refer the reader to  surveys \cite{Nishimura18,Heuvel13ComChng}.

Three models 
have been considered for vertex cover reconfiguration. The first is the Token Addition and Removal (TAR) model, which is the one we adopt in this paper,
where we can move from an intermediate solution $S_1$ to another one, $S_2$, if they differ by a single vertex. 
In the other two models, Token Sliding (TS) and Token Jumping (TJ), 
two intermediate solutions $S_1$ and $S_2$ are adjacent if $S_2$ is obtained from $S_1$ by swapping a vertex $v$ in $S_1$ with another one, $u$, which is arbitrary in TJ, but must be a neighbor of $v$ in TS.

The trend of hardness persists in vertex cover reconfiguration too. Here, the decision problem is: given two vertex covers of size at most $k$, is there a reconfiguration schedule transforming one into the other
via TAR, 
and without ever having a vertex cover of size more than $k+1$. The problem is \PSPACE-complete even in usually tractable graph classes such as perfect graphs~\cite{KaminskiMM12}, graphs of bounded  pathwidth or bandwidth~\cite{Wrochna18}, and planar graphs of degree at most 3~\cite{HearnD05}. It is 
also 
NP-complete for bipartite graphs~\cite{LokshtanovM19}. 
We refer the reader to ~\cite[Figure 5]{Nishimura18} for a  graphic depiction of the complexity landscape of the problem. 

On the positive side, the vertex cover
reconfiguration schedule existence problem is known to be polynomially solvable for trees and cactus graphs~\cite{Mouawad_2018,BonsmaKW14,DEMAINE2015,HoangU16,KaminskiMM12}. More relevant to our work, it is known that there are approximation-style schedules in cactus graphs, such that 
in all intermediate solutions, the size of the vertex cover is bounded by the larger of the two initial ones, plus 2 \cite{Mouawad_2018,KaminskiMM12}. While seemingly quite specialized, these results give schedules with similar guarantees for wider classes of graphs, such as \emph{even-hole-free graphs} (that is, graphs without an induced even cycle), which include \emph{chordal graphs} and \emph{cographs}.
In the same spirit, it is  shown in \cite{deBergJM16} that there are schedules in general graphs that increase the maximum vertex cover by at most the \emph{pathwidth} of the graph. All these schedules 
are \emph{monotone}, i.e., every vertex changes its status at most once. This is the case with all schedules in our paper as well. In fact, as it is shown in~\cite{LokshtanovM19}, when the two vertex covers in a reconfiguration instance are \emph{disjoint}, there always exists a monotone schedule. 

Finally, we note that parallelism in reconfiguration has been considered in settings other than vertex cover reconfiguration, e.g., \cite{KawaharaSY17}. Distributed reconfiguration 
has 
also
been studied for colorings \cite{BonamyORSU18} and independent sets \cite{Censor-HillelR19}.

\subparagraph*{Roadmap.}  
In \cref{sec:defs} 
we provide the basic definitions. In \cref{sec:centralized} we prove \Cref{thm::sched_compres} (batch-compression) and present existential results on batch and non-batch reconfiguration schedules. In \Cref{sec:tecDist} we formally introduce the concept of small separator decompositions, prove \Cref{gen::lma::decomp_to_sched_computation}, and we present all core steps to prove \Cref{thm:cactus_core_graphs}. 
All remaining results and missing proofs  appear in the appendix.

\section{Problem Statement, Definitions,  and Notation}
\label{sec:defs}

A \emph{vertex cover} in a graph $G=(V,E)$ is a subset $S\subseteq V$ such that every edge has at least one of its end-vertices in  $S$. An \emph{independent set} is a subset $S\subseteq V$ of vertices such that every edge has at most one of its end-vertices in $S$.
Note that if $S$ is a vertex cover then $V\setminus S$ is an independent set, and vice versa. We use the notation $[x]=\{0,1,\dots,x\}$, for an integer $x\ge 0$.
For sets $A, B$, we let $A\oplus B=(A\setminus B) \cup (B\setminus A)$ denote their \emph{symmetric difference}.
\begin{definition}[Vertex Cover Reconfiguration Schedule]
\label{def:reconfigurationSchedule}
Given a graph $G = (V,E)$ and two vertex covers $\alpha,\beta\subseteq V$ 
(not necessarily minimal), a \emph{reconfiguration schedule} $\mathcal{S}$ from $\alpha$ to $\beta$ of length $\ell$ and \emph{cost} $s$ is a sequence $(V_i)_{i\in [\ell]}$ of vertex covers of $G$ such that 
\begin{compactenum}
    \item $V_0 = \alpha$ and $V_\ell = \beta$,
    \item $\forall i\in [\ell-1], \quad |V_i\cup V_{i+1}| \leq s$, 
    \item $\forall i \in [\ell-1], \quad V_i \oplus V_{i+1}$ is an independent set of $G$.
\end{compactenum}
\end{definition}
The sets $\mathcal{E}_i:=V_{i}\oplus V_{i+1}$, which we call \emph{batches}, contain the vertices that are added to or removed from the current vertex cover in order to obtain the next one.
Thus, we also refer to a length-$\ell$ schedule as an \emph{$\ell$-batch schedule}. A reconfiguration schedule is 
\emph{sequential} 
if $|\mathcal{E}_i| = 1$, for all $i\in[\ell-1]$,
and a \emph{batch} schedule otherwise. 
A schedule $\mathcal{S}$ from $\alpha$ to $\beta$  is \emph{monotone} if every node $v \in V$ is changed at most once in the schedule.
Note that a monotone schedule $\mathcal{S}$ only changes nodes in $\alpha \oplus \beta$, and therefore its length (if there are no empty batches) is upper bounded by $|\alpha \oplus \beta|$. In particular, 
a monotone schedule does not change nodes in $\alpha\cap\beta$. 

The first property of \Cref{def:reconfigurationSchedule} ensures that $\mathcal{S}$ is a reconfiguration schedule from $\alpha$ to $\beta$. The second and third properties together ensure robustness, in the sense that not only $(V_i)_{i\in [\ell]}$ are vertex covers of size at most $s$, but any reconfiguration sequence between $V_i$ and $V_{i+1}$ yields vertex covers with the claimed size bound. Note that when in each reconfiguration step we only add or remove nodes, property $2$ reduces to having $|V_i| \leq s$ for all $i \in [\ell]$.

For $\eta,c\in \mathbb{R}_+$, a reconfiguration schedule from a vertex cover $\alpha$ to a vertex cover $\beta$ is an \emph{$(\eta, c)$-approximation} if its cost is at most $\eta\max\{|\alpha|, |\beta|\}+c$. 
Note that here we do not compare to the cost of an optimal reconfiguration schedule but rather to $M=\max\{|\alpha|, |\beta|\}$; the cost of an optimal schedule is always at least $M$, but often it is larger than $M$.

\begin{observation}[Reverse Schedule]\label{obs:reverseschedule}
Let $\mathcal{S}=(V_i)_{i\in [\ell]}$ be a vertex cover reconfiguration schedule of cost $s$ from $\alpha$ to $\beta$ in graph $G$. Then the \emph{reverse schedule} $\mathcal{S'}=(V'_i)_{i\in [\ell]}$, where $V'_i=V_{\ell-i}$, is a vertex cover reconfiguration schedule of cost $s$ and length $\ell$ from $\beta$ to $\alpha$.
\end{observation}

\subparagraph*{Special graph classes.}
A graph is \emph{even-hole-free} if it contains no induced cycle with an even number of vertices. A graph is \emph{chordal} if it contains no induced cycle with 4 or more vertices  
(in particular, it is even-hole-free).
A graph is \emph{claw-free} if it contains no induced $K_{1,3}$ sub-graph. A graph is a \emph{cactus} graph if each of its edges belongs to at most $1$ cycle. Alternatively, cactus graphs are characterized by a forbidden minor, the diamond graph, which is obtained by removing an edge from $K_4$; thus, they form a minor-closed family.

\subparagraph*{The LOCAL model of distributed computing~\cite{linial92, peleg2000distributed}.}
A connected graph is abstracted as an $n$-node network $G=(V, E)$. 
Communications happen in synchronous rounds. Per round, each node can send one (unbounded size) message to each of its neighbors. Further, each vertex has a unique ID from a space of size $\poly~n$.
We say that the nodes of graph $G$ distributively compute a reconfiguration schedule $\mathcal{S}=(V_i)_{i\in [\ell]}$ if each node $v \in V$ knows its membership in each $V_i$. We emphasize that the length of a schedule, i.e., the number of batches, 
and the distributed time to compute the schedule are separate measures.

\subparagraph*{Parameters and notation.}
Let $G=(V,E)$ be a graph with two vertex covers $\alpha$ and $\beta$. For a subset 
$U\subseteq V$, we let $G[U]$ denote the graph induced by $U$.  We use $n=|V|$ for the number of vertices, $\Delta$ for the maximum degree, and $\lambda$ for the arboricity of $G$. We 
use the notation $M=\max\{|\alpha|, |\beta|\}$, $m=\min\{|\alpha|, |\beta|\}$, $D=\alpha\oplus\beta$, and $X=\alpha\cap \beta$. Note that 
$G[D]$ is \emph{bipartite}, $D\cup X=\alpha\cup \beta$, and that nodes in $V\setminus (D\cup X)$ form an \emph{independent set}.

\section{Sequential and Batch Reconfiguration Schedules}
\label{sec:centralized}

Our first result of this section is based on the observation that any $k$ (monotone) sequential reconfiguration steps can be replaced by two batches: 
the first adds all $\beta$-nodes that appear in this sequence, and the second removes all $\alpha$-nodes that appear in it. 
In the following theorem, we use a similar (but more careful) 
manipulation of schedules 
%and give a formal meaning to the notion of 
to \emph{compress} a long sequential schedule into a short batch schedule, while bounding the overhead in approximation. Its tradeoff is asymptotically tight (cf. \Cref{thm::t_lb}).

\schedCompress*

\begin{proof}
If $|\beta\setminus\alpha|=0$, the claim immediately follows by removing all vertices in $|\alpha\setminus\beta|$ in a single batch. Otherwise, note that the schedule $\mathcal{S}$
only reconfigures vertices in $\alpha\oplus\beta$ and only adds/removes each vertex once. Let $\mathcal{S}_{\alpha}=(u_i)_{i=1}^{|\alpha\setminus\beta|}$ and $\mathcal{S}_{\beta}=(v_i)_{i=1}^{|\beta\setminus\alpha|}$ be schedule $\mathcal{S}$ restricted to vertices of $\alpha$ and $\beta$, respectively, without changing their order (formally, $\mathcal{S}_{\alpha}$ and $\mathcal{S}_{\beta}$ consist of the vertices in the singleton batches of $\mathcal{S}$).  Define three positive integers $$s=\ceil*{\eps |\beta\setminus \alpha|}, \text{  } r=\floor*{\eta M} + c - |\alpha|+s,\text{and } \ell= \ceil*{(\max\{|\beta\setminus\alpha|-r,0\})/s}.$$

The \textbf{schedule $\mathcal{S}'$} consists of $2\ell + 2$ batches $\mathcal{E}_0, \mathcal{E}_1, \dots, \mathcal{E}_{2\ell+1}$, where in the first batch we add the first $r$ vertices from $\beta\setminus\alpha$, according to $\mathcal{S}_{\beta}$, then in every subsequent pair of batches we remove $s$ vertices from $\alpha\setminus\beta$, then add $s$ vertices from $\beta\setminus\alpha$, according to $\mathcal{S}_{\alpha}$ and $\mathcal{S}_{\beta}$. In the final batch $\mathcal{E}_{2\ell+1}$, we remove the remaining vertices, if any, from $\alpha\setminus\beta$.
We introduce notation to define the schedule formally. Given $i<j$, denote by $[u_i,u_j]$ and $[v_i,v_j]$ the sets $\{u_i,u_{i+1},\hdots,u_j\}$ and  $\{v_i,v_{i+1},\hdots,v_j\}$, respectively, with the convention $[u_i,u_j]=\emptyset$ if $i>|\alpha\setminus \beta|$ (resp. $[v_i,v_j]=\emptyset$ if $i>|\beta\setminus \alpha|$), and $[u_i,u_j]=[u_i,u_{|\alpha\setminus\beta|}]$ if $j>|\alpha\setminus \beta|$ (resp. $[u_i,u_j]=[u_i,u_{|\beta\setminus\alpha|}]$ if $j>|\beta\setminus \alpha|$).
Define  $\mathcal{E}_0=\brackets*{v_1,v_r}$, $\mathcal{E}_{2i-1}=\brackets*{u_{(i-1)s+1}, u_{is}}$,  $\mathcal{E}_{2i}=\brackets*{v_{r+(i-1)s+1}, v_{r+is}}$, for $i=1,2,\dots,\ell$, and $\mathcal{E}_{2\ell+1}=\brackets*{u_{\ell s+1}, u_{|\alpha\setminus \beta|}}$.

\smallskip
The schedule $\mathcal{S}'$ is monotone, as each node appears in $\mathcal{S}'$ the same amount of times it appears in $\mathcal{S}$.
We next show that we add all vertices of $\beta\setminus\alpha$ in $\mathcal{S}'$. In batch $\mathcal{E}_0$ we add $r$ vertices of $\beta\setminus\alpha$, and if $\ell > 1$ (which doesn't happen only if $r\geq |\beta\setminus\alpha|$), in batches $\mathcal{E}_2,\mathcal{E}_4\dots,\mathcal{E}_{2\ell}$ we add $s$ vertices per batch. In total, these are $r+\ell\cdot s\geq |\beta\setminus \alpha|$ vertices.

With $r\geq s$ and by the definition of $\ell$, we upper-bound the number of batches of $\mathcal{S}'$ by either $2\leq 2\ceil{1/\eps}$ (if $\ell = 1$) or by
$
    2\ell + 2 = 2\ceil*{\frac{|\beta\setminus\alpha|-r}{s}}+2 \stackrel{r\geq s}{\leq} 2\ceil*{\frac{|\beta\setminus\alpha|}{s}} \stackrel{s=\lceil \eps|\beta\setminus\alpha|\rceil}{\leq} 2\ceil*{\frac{1}{\eps}}.
$

\textbf{Approximation.} 
After the first batch, the vertex cover is of size $|\alpha|+r = \floor*{\eta M} + c+s \le (\eta+\eps) M + c+1$. Afterwards, as long as we remove $s$ $\alpha$-nodes and then add $s$ $\beta$-nodes in each batch, the size of the vertex cover never goes above $(\eta+\eps) M + c+1$. If we add fewer than $s$ $\beta$-nodes but remove $s$ $\alpha$-nodes, the size of the vertex cover decreases. If we, in some batch, remove less than $s$ $\alpha$-nodes, then we have removed all vertices in $\alpha\setminus \beta$, and the output is a subset of $\beta$, that is, we trivially satisfy the approximation guarantee.

\textbf{Validity.} Lastly, we show that all edges are covered after each reconfiguration step and that each batch is an independent set. Assume, for contradiction, that $i$ is an index such that, in schedule $\mathcal{S}'$, node $u_i\in\mathcal{S}_{\alpha}$ is removed and it has a neighbor $v_j\in\mathcal{S}_{\beta}$ that has not yet been added to the vertex cover in $\mathcal{S}'$ (in particular, $v_j$ could be in the same batch with $u_i$). This implies that $j \geq r + i - (s-1)$, and that in $\mathcal{S}$, right after the addition of the $j$-th $\beta$-node, the $i$-th $\alpha$-node has not been removed. Thus, $j$ nodes have been added, and at most $i-1$ nodes have been removed, i.e., the size of the vertex cover is 
\[
    |\alpha| + j - (i-1) \geq |\alpha| + r + i - s+1 - i+1 = |\alpha|+r-s+2=\floor*{\eta M} + c + 2 > \eta M + c\ .
\]
This is a contradiction to the assumption that $\mathcal{S}$ is an $(\eta,c)$-approximation schedule. Hence, for a node $v \in \mathcal{E}_{i'} \cap \alpha$, all of its neighbors must be in $\mathcal{E}_{j'}$'s with $j'<i'$, implying that every set in $\mathcal{S'}$ is indeed a vertex cover. This, in particular, shows that for every $u \in \mathcal{E}_{i'
} \cap \alpha$ and $v \in \mathcal{E}_{i'} \cap \beta$, it holds that $(u,v) \notin E$. Together with the fact that $\mathcal{E}_{i'} \cap \alpha\setminus\beta$ and $\mathcal{E}_{i'} \cap \beta\setminus\alpha$ are independent sets, this shows that every $\mathcal{E}_{i'}$ is an independent set, for all $i'\in[\ell]$. 
\end{proof}

To apply \Cref{thm::sched_compres}, we need sequential schedules. The following theorem summarizes some results for planar and cactus graphs that are either known or follow from known results~\cite{ItoNZ16,deBergJM16}. 
%For trees, the additive 2 can be replaced with additive 1~\cite{ItoNZ16,Mouawad_2018}.
The result for planar graphs  further extends to  other graph classes that have small separators, e.g., \emph{fixed minor-free} graphs.
\begin{restatable}[Sequential Schedules]{theorem}{seqSched} 
\label{thm:sequential}
Let $G$ be a graph with vertex covers $\alpha,\beta$. If $G$ is a cactus graph,  
it admits
a monotone $(1,2)$-approximation schedule.
If $G$ is a planar graph, 
it admits
a monotone $(1,O(\sqrt{M}))$-approximation schedule,
where $M=\max\{|\alpha|,|\beta|\}$.
\end{restatable}

We provide a very simple unified proof of these results. 
The centerpiece of our proof is the following lemma about graphs with certain ``well-behaved'' separation of vertices. It will also be used for distributed computation of schedules (cf. \cref{gen::lma::decomp_to_sched}).

\begin{restatable}{lemma}{lemdiscon}\label{lem:discon}
Let $G=(V_1\mathbin{\dot{\cup}} V_2, E)$ be a graph with vertex covers $\alpha,\beta$, such that there is no edge between $V_1\cap\alpha$ and $V_2$. Let $\alpha_i=V_i\cap \alpha$, $\beta_i=V_i\cap \beta$. Assume that $|\beta_1|-|\alpha_1|\le |\beta_2|-|\alpha_2|+1$. If, for some $\eta\in [1,2]$, $k\ge 0$, $\mathcal{S}_1$ and $\mathcal{S}_2$ are $(\eta,k)$-approximation schedules for $G_1,\alpha_1,\beta_1$ and $G_2,\alpha_2,\beta_2$ (resp.), then their concatenation $\mathcal{S}_1,\mathcal{S}_2$ is an $(\eta,k)$-approximation schedule for $G$.  
\end{restatable}
\begin{proof}
	First, observe that $\mathcal{S}_1,\mathcal{S}_2$ is indeed a valid schedule, since, as assumed, there is no $\alpha$-node in $V_1$ that depends on a $\beta$-node in $V_2$. Let $a_i=|\alpha_i|$, $b_i=|\beta_i|$. 
	Consider the size of the vertex cover during the reconfiguration of $G_1$. By the assumption, it is at most $\eta\max(a_1,b_1)+k+a_2=\eta \max(|\alpha|,|\alpha|+b_1-a_1)+k$, using $|\alpha|=a_1+a_2$, $\eta\ge 1$. If $b_1-a_1\le 0$, the size is at most $\eta|\alpha|+k$, as required. Otherwise, we have $b_1-a_1\ge 1$, so $\eta\max(a_1,b_1)+a_2+k\le\eta (b_1+a_2)+k\le\eta(|\beta|+a_2-b_2)+k\le \eta|\beta|+k$, where we used $|\beta|=b_1+b_2$, and $b_2-a_2\ge b_1-a_1-1\ge 0$. Next, suppose we continue with reconfiguring $G_2$. If we run this schedule backwards, we get a $\beta$-to-$\alpha$ reconfiguration schedule, starting with $G_2$ (cf. \Cref{obs:reverseschedule}). Renaming $\beta\leftrightarrow\alpha$, $\beta_i\leftrightarrow\alpha_i$, $\mathcal{S}_i\leftrightarrow reverse(\mathcal{S}_{3-i})$, $G_1\leftrightarrow G_2$, the lemma conditions still apply, and the same analysis above shows that the size of the vertex cover is below $\eta\max(|\alpha|,|\beta|)+k$ when processing $G_2$ as well.  
\end{proof}

The proof of \cref{thm:sequential} for cactus graphs follows by an inductive application of the lemma, and using the fact that cactus graphs form a hereditary class, and every connected cactus graph is either a cycle graph or has a cut-vertex (which makes the lemma immediately applicable). The proof for planar graphs is also an inductive application of the lemma, relying on the well-known fact that planar graphs form a hereditary class, and each $n$-vertex planar graph has a separator of size $O(\sqrt{n})$ \cite{LiptonT80}. See \cref{app:sequential} for details.

The theorem above applies not only to the mentioned  classes, but also to graphs $G$ for which $G[\alpha\oplus\beta]$ belongs to those classes. We thus get the following batch scheduling theorem.

\begin{restatable}{observation}{ehFree} 
\label{obs:eh_free}\label{obs:claw_free}
Let $G = (V,E)$ be either a chordal graph, an even-hole-free graph or a claw-free graph with vertex covers $\alpha$ and $\beta$, then $G[\alpha\oplus\beta]$ is a cactus graph.
\end{restatable}

\begin{proof}
	Let $D=\alpha\oplus\beta$. Note that $\alpha\setminus\beta$ and $\beta\setminus\alpha$ are independent sets (as a complement of a vertex cover), and so, $G[D]$ is bipartite and does not contain odd cycles.
	
	If $G$ is even-hole free (which includes chordal graphs), then $G[D]$ contains neither even nor odd (as mentioned above) cycles, and therefore it is a forest.  %First, assume that $G$ is even-hole-free (which includes chordal graphs). Since $G[D]$ is bipartite, it has no odd cycles, and since it is also even-hole-free, it contains no even cycles either. Therefore $G[D]$  is a forest. 
	
	Next, assume that $G$ is a claw-free graph. If $G[D]$ has a node $v$ with at least 3 neighbors $u_1,u_2,u_3$, then there is no edge between $u_1,u_2,u_3$ (no odd cycles), which contradicts to $G$ being claw-free. % and assume that there exists a node $v$ with degree at least $3$ with neighbors $u_1, u_2, u_3$. The graph induced by $\{v,u_1,u_2,u_3\}$ is a claw (since there are no edges between the $u_i$ nodes as $G[D]$ is bipartite), which is a contradiction to $G$ being claw-free. 
	Therefore, every node in $D$ has a degree at most $2$, and this shows that each connected component is either a path or a cycle.
\end{proof}

\variousClassesBatch*
\begin{proof}
By \Cref{obs:eh_free}, 
for every graph in a class mentioned in the first claim, 
$G[\alpha\oplus\beta]$ 
is a cactus graph, which by \Cref{thm:sequential}, has a monotone $(1,2)$-approximation schedule, and by \Cref{thm::sched_compres}, has a $(2\lceil 1/\eps\rceil)$-batch $(1+\eps,3)$-approximation schedule. For the second claim, let $M=\max\{|\alpha|,|\beta|\}$ and $\delta=\eps/2$. By \Cref{thm:sequential}, there is a monotone schedule with a worst-case vertex cover of size $M+O(\sqrt{M})=(1+\delta)M+\sqrt{M}(O(1)-\delta\sqrt{M}/2)$. If $M\ge c/\delta^2$, for a large enough constant $c$, then the additive term is negative. Otherwise, it is $O(1/\delta)$. Thus, we have a sequential $(1+\delta,O(1/\delta))$-approximation schedule. \Cref{thm::sched_compres} with $\eps=\delta$ 
then gives a $(2\lceil 2/\eps\rceil)$-batch $(1+\eps,O(1/\eps))$-approximation schedule, as claimed.
\end{proof}

%Using an application of the pigeonhole principle, 
\Cref{thm::sched_compres,thm:variousClassesBatch} are nearly optimal, in the sense that a general compression theorem cannot prove a better trade-off between approximation overhead and schedule length. This holds even for path graphs.

\begin{restatable}[Batch Lower Bound]{theorem}{batchLowerBound}
\label{thm::t_lb}
For any $1 \leq t < n$, there are two vertex covers  $\alpha$, $\beta$ of a $2n$-vertex path graph such that every $(t+1)$-batch reconfiguration schedule from $\alpha$ to $\beta$ is no better than a $\parenths*{1+\frac{1}{t},0}$-approximation.
\end{restatable}
\begin{proof}
	Consider a path graph on $2n$ vertices, represented as a bipartite graph $G=(A\cup B, E)$ with $|A|=|B|=n$. Observe that both $A$ and $B$ are minimal vertex covers in $G$. Consider any $(t+1)$-batch reconfiguration schedule $\mathcal{S}$ with batches $\mathcal{E}_0,\mathcal{E}_1,\dots,\mathcal{E}_{t}$ from $A$ to $B$. It is easy to translate $\mathcal{S}$ into a refined schedule with
	the same approximation ratio and
	at most $2t$ batches, $\mathcal{E}'_0,\mathcal{E}'_1,\dots,\mathcal{E}'_{2t}$, such that for each $i$, $\mathcal{E}'_i\subseteq A$ or $\mathcal{E}'_i\subseteq B$. Let $\mathcal{E}'_{i_1},\dots,\mathcal{E}'_{i_{t}}$ be the batches that are contained in $B$. Since $\bigcup_{j=1}^{t}\mathcal{E}'_{i_j}=B$, there is an index $j$, such that $|\mathcal{E}'_{i_j}|\ge |B|/t=n/t$. Note that before adding $\mathcal{E}'_{i_j}$ to the current vertex cover, the vertex cover had size at least $n$ (the size of a minimum vertex cover), hence after adding it, we have a vertex cover of size at least $n(1+1/t)$, which proves the claim. 
\end{proof}

\section{Distributed Computation of Schedules}
\label{sec:tecDist}

The main objective of this section is to show the following result. 

\thmCactusCoreRecon*

More concretely, the result holds for the following graph classes.

\begin{restatable}{corollary}{corfastDistributedReconf}\label{cor:variousClassesDist} 
\begin{sloppypar}For each $\eps>0$, there is a deterministic \LOCAL algorithm with round complexity $O(\logstar(n)/\eps)$ that for any connected $n$-node graph $G=(V,E)$ with vertex covers $\alpha$ and $\beta$, computes an $O(1/\eps)$-batch $(1+\eps,3)$-approximation schedule, if $G$ belongs to one of the following graph classes $\{\mathsf{cycles, trees, forests, chordal, cactus, even\text{-}hole\text{-}free, claw\text{-}free}\}$.
\end{sloppypar}
\end{restatable}

To efficiently compute a reconfiguration schedule in a distributed setting, we want to deal with large parts of the graph at the same time, e.g., by decomposing the graph into independent/non-interfering clusters. To this end, we define the concept of  \emph{small separator decompositions}. These are sufficient to compute \emph{good} reconfiguration schedules, both in terms of  the schedule length and  approximation factor, as well as the computation time. Afterwards, we show that such decompositions can be computed 
efficiently, i.e., in $O(\logstar(n)/\eps)$ rounds, on several graph classes, such as \emph{cactus} and \emph{even-hold-free} graphs. In fact, it is sufficient if $G[\alpha\oplus\beta]$ falls into the respective graph class.

\subsection{Small Separator Decompositions}
We first present the formal definition of a small separator decomposition and show how to use it to compute batch reconfiguration schedules. See \Cref{fig:my_label} for an illustration.

\begin{figure}[t!]
    \centering
    \includegraphics[height=2.3cm]{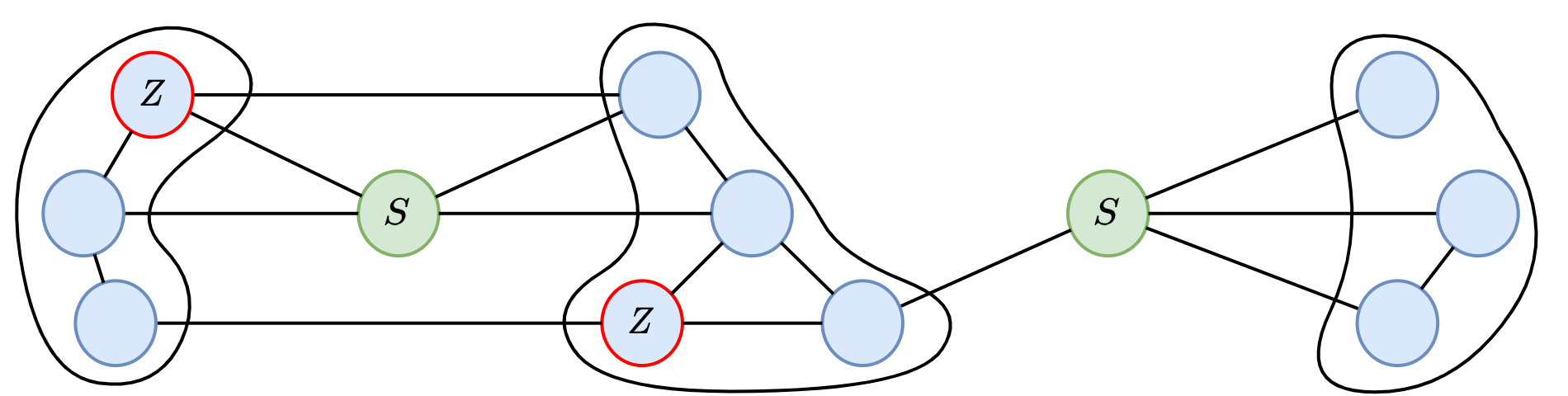}
    \caption{$(d,x)$-decomposition with $d=x=2$ and with $k=3$ clusters (one disconnected). Every node is either clustered or belongs to $S$, and every edge is either in a cluster or intersects $S\cup Z$.}
    \label{fig:my_label}
\end{figure}

\begin{restatable}[Small Separator Decomposition]{definition}{DefsmallSeparator} \label{def:smallSeparator}
Let  $G=(V,E)$ be a graph with a subset $Z\subseteq V$ of \emph{ghost nodes}. 
A collection of \emph{clusters} $C_1,\dots,C_k\subseteq V$ and a \emph{separator set} $S \subseteq V$ is a \emph{(weak) $(d,x)$-separator decomposition} with regard to $Z$ if the following hold:
\begin{compactenum}
		\item partition: $C_1,\dots,C_k, S$ forms a partition of $V$, 
    \item the weak diameter of each $C_i$
    (i.e. $\max\{dist_G(u,v) | u,v\in C_i\}$)
    is upper bounded by $d$,
		\item small separator set: $|S|\leq x$,
		\item $S$ separates: There is \emph{no} edge $\{u,v\}\in E$ with $u\in C_i\setminus Z$, $v\in C_j\setminus Z$, and $i\neq j$. 
\end{compactenum}
		When additionally $G[C_i]$ has diameter at most $d$, for all $1\leq i \leq k$, we speak of a \emph{strong $(d,x)$-separator decomposition}. 
		A \emph{$d$-diameter clustering} $\mathcal{C}=\{C_i\}_{i=1}^k$ of a graph $G=(V,E)$ is a partitioning $V=\bigcup_{i=1}^k C_i$ of the vertex set into $k$ disjoint subsets, for a parameter $k\ge 1$, such that for every $1\le i\le k$, $G[C_i]$ is connected and has strong diameter at most $d$. 
\end{restatable}
Given a graph $G$ and vertex covers $\alpha$ and $\beta$, we would like to compute a monotone schedule for each cluster independently. To do so without interference, the separator  set $S$ is needed. To save on the size of $S$, we observe that nodes in $V\setminus (\alpha\oplus\beta)$ never change, and so, can be ignored and are referred to as \emph{ghost} nodes.
Next, we show that such
decomposition are helpful to compute batch reconfiguration schedules.

\begin{restatable}{lemma}{lemseparatorHelps}
\label{gen::lma::decomp_to_sched}
Let $G=(V,E)$ be a graph with two vertex covers $\alpha$ and $\beta$ and a $(d,x)$-separator decomposition $C_1,\dots,C_k, S$  with ghost nodes $Z\subseteq V\setminus (\alpha\oplus\beta)$. If each $G[C_i]$ admits a monotone $\ell$-batch  $(\eta,c)$-approximation schedule $\mathcal{S}_i$,
then $G$ admits
a $(2\ell+2)$-batch $(\eta,x+k\cdot c)$-approximation schedule $\mathcal{S}$.
\end{restatable}

\begin{proof}
Let $I_{\alpha}=\{i : |C_i \cap \alpha| > |C_i \cap \beta|\}$ be the indices of  clusters $C_i$ with more $\alpha$-nodes than $\beta$-nodes, %the family of all the sets $C_i$ for which $|C_i \cap \alpha| > |C_i \cap \beta|$, 
and let
$I_{\beta}$ be the indices of the rest
(for which $|C_i \cap \alpha| \leq |C_i \cap \beta|$). Consider the monotone schedule $\mathcal{S}_\alpha$ for $G_\alpha=G[\cup_{I_\alpha} C_i]$ that consists of a parallel execution of $\mathcal{S}_i$, for all $i\in I_\alpha$:  the $t$-th batch  of $\mathcal{S}_\alpha$ is the union of $t$-th batches of $\mathcal{S}_i$, $i\in I_\alpha$. Observe that $\mathcal{S}_\alpha$ is a $(\eta, kc)$-approximation for $G_\alpha$. Validity follows from the fact that all edges between different $C_i$ are adjacent to ghost nodes $Z\subseteq V\setminus (\alpha\oplus \beta)$, which are never touched by the \emph{monotone} schedule $\mathcal{S}_\alpha$. For the approximation, note that at each step $t$, within each $C_i$, the current vertex cover is bounded by $\eta|C_i\cap \alpha|+c$, by the definition of $\mathcal{S}_i$ and $I_\alpha$, hence the  size of the current vertex cover in $G_\alpha$ at step $t$ is at most  $\eta\cdot \sum_{I_\alpha} |\alpha\cap C_i|+k\cdot c\le \eta |\alpha|+k c$. By a symmetric construction, we get a monotone $(\eta, kc)$-approximation schedule $\mathcal{S}_\beta$ for $G_\beta=G[\cup_{I_\beta}C_i]$. By \cref{lem:discon}, $\mathcal{S}_\alpha, \mathcal{S}_\beta$ is a $(\eta, kc)$-approximation schedule for $G_\alpha\cup G_\beta=G\setminus S$ (again, we may ignore the edges between $C_i$, since they are covered by ghost nodes). The schedule for $G$ is then $\mathcal{S}=(S\cap(\beta\setminus \alpha)), \mathcal{S}_\alpha, \mathcal{S}_\beta, (S\cap (\alpha\setminus \beta))$, where we handle the separator set $S$ in a na\"ive way. By the properties of $\mathcal{S}_\alpha$ and $\mathcal{S}_\beta$ and by \cref{lem:discon}, $\mathcal{S}$ is a $(2\ell+2)$-batch $(\eta, kc+x)$-approximation schedule. 
\end{proof}

Given a $(d,x)$-separator decomposition with $k$ clusters, all phases in the proof of \Cref{gen::lma::decomp_to_sched} can be executed in $O(d)$ rounds in the \LOCAL model, and we obtain the following theorem.

\begin{restatable}{theorem}{lemdecomptoschedcomputation}
\label{gen::lma::decomp_to_sched_computation}
If the conditions of \cref{gen::lma::decomp_to_sched} hold,
then a $(2\ell+2)$-batch $(\eta,x+k\cdot c)$-approximation schedule for $G$
can be computed in $O(d)$ rounds in the \LOCAL model. 
\end{restatable}

There are three parameters at play in when using small separator decomposition to compute  schedules: $x$ -- the separator set size, $k$ -- the number of clusters (both affect the approximation factor of the 
schedule), and $d$ --  the cluster diameter (which affects the time it takes to compute the schedule). 
In the next section, we show how to balance those three parameters on a graph $G$ for which $G[\alpha \oplus \beta]$ is a cactus graph. In addition, recall that if $G[\alpha\oplus\beta]$ is a cactus graph, our batch compression results from \Cref{thm:variousClassesBatch} imply the existence of good monotone schedules for any subgraph, as needed by \Cref{gen::lma::decomp_to_sched,gen::lma::decomp_to_sched_computation}.

\subsection{Computing Small Separator Decompositions}

To prove \Cref{thm:cactus_core_graphs}, we show how to compute a small separator decomposition with parameters $d = O(1/\eps)$, $x = \eps |\alpha \oplus \beta|$ and $k = \max\{1,2\eps M\}$ (where $M = \max\{\alpha,\beta\}$) on graphs $G$ for which $G[\alpha\oplus\beta]$ is a cactus graph. We stress that the usual approaches for computing (small diameter) decompositions, such as \cite{awerbuch89,linial93,Vaclav_2019,SLOCAL17}, neither include any vertex-separators nor give any promise on the number of clusters. Nevertheless, in \cref{sec:SLOCAL}, we show how to use these methods to compute decompositions for general graphs in $\poly\log n$ rounds, deterministically. Using the results about batch schedules for bounded arboricity graphs, 
this implies that we can also compute such schedules in the distributed setting (see \Cref{thm:distrPlanar}). A downside is that those techniques inherently require $\Omega(\log n)$ rounds and are not useful for proving \Cref{thm:cactus_core_graphs}. Therefore, new methods for computing our decomposition are needed to prove the following result.
 
\begin{restatable}[Cactus-Core Decomposition]{lemma}{lemCactusCoreDecomp}\label{lem:cactus_core_decomp} For each $\eps>0$, there is a \LOCAL algorithm that for any connected $n$-node graph $G=(V,E)$ and vertex covers $\alpha, \beta$ such that $G[\alpha\oplus\beta]$ is a cactus graph,
computes a $(d,x)$-separator decomposition with $d = O(1/\eps)$, $x = \eps |\alpha \oplus \beta|$, a ghost node set $Z = V \setminus (\alpha \oplus\beta)$ and $\max\{1,2\eps M\}$  clusters in $O(\logstar(n)/\eps)$ rounds.
\end{restatable}

\Cref{thm:cactus_core_graphs} (formal proof in \cref{sec:distributed}) follows by plugging the results of \Cref{thm:variousClassesBatch} and \Cref{lem:cactus_core_decomp}
into  \Cref{gen::lma::decomp_to_sched_computation}.
In the proof of the next lemma (see \cref{sec:distributed}), we merge small diameter clusters until all clusters have diameter of at least $1/\eps$. The cluster merging is inspired by the classic distributed minimum spanning tree algorithms in \cite{GallagerHS83_MST,GarayKP98_MST}.

\begin{restatable}[Cluster Merging]{lemma}{lemClusterMerging} \label{lem:compressingClustering}
	Let  $G=(V,E)$ be a graph with a $d$-diameter clustering $\mathcal{C}= (C_i)_{i\in k}$.  
	For every $\eps\in (0,1)$, there is an  $O\parenths*{\log^* (n)/\eps + d}$-round \LOCAL algorithm that 
	using $\mathcal{C}$, computes a $\parenths*{O(1/\eps)+d}$-diameter clustering $\mathcal{C'}$, where additionally, each cluster either has diameter at least $1/\eps$ or consists of a whole connected component of $G$.
\end{restatable}

%Due to the technicalities and the handling of ghost nodes and all parameters, 
The (somewhat technical)  proof of \Cref{lem:cactus_core_decomp} is in \cref{sec:distributed}. We sketch it below. 
 
\begin{proof}[Proof sketch for \Cref{lem:cactus_core_decomp}]
In short, to compute the desired decomposition, we begin with the trivial decomposition, where each node of $G$ forms its own $0$-diameter cluster, and the trivial (but large) separator set $S=\alpha\oplus\beta$; the number of clusters is 
upper bounded by $n$.
Next, in two steps, we first merge the clusters of $G[\alpha\oplus\beta]$ via \Cref{lem:compressingClustering}, in order to reduce the number of separating nodes, and then we merge all clusters of $G$ via \Cref{lem:compressingClustering}, in order to reduce the number of clusters. We detail on the crucial ingredients for the whole process. 

\textbf{$(d,x)$-separator decomposition $\mathcal{C}_1$ with $d = 0$ and $x =  |\alpha\oplus\beta|$, $k=n$:}
Each node of $G$ forms its own $0$-diameter cluster and we use  separator set $S = \alpha\oplus\beta$.

\textbf{$(d,x)$-separator decomposition $\mathcal{C}_2$ with $d = O(1/\eps)$ and $x = \eps |\alpha\oplus\beta|$, $k=n$:}
We consider the induced graph $G[\alpha\oplus\beta]$ and apply \cref{lem:compressingClustering} on $\mathcal{C}_1$ restricted to $G[\alpha\oplus\beta]$, where every node in $S$ is treated as a $0$-diameter cluster, to compute the clusters of the $O(1/\eps)$-diameter clustering $\mathcal{C}_2$, in which each cluster either contains at least $1/\eps$ nodes, or a complete connected component of $G[\alpha\oplus\beta]$. If we ignore (for the sake of simplifying this proof sketch) the clusters of $\mathcal{C}_2$ that form complete connected components, we can upper bound the number of clusters by  $x=O(\eps |\alpha\oplus\beta|)$  (as each cluster contains at least $1/\eps$ nodes).  Since cactus graphs are closed under minor operations and any $x$-node cactus graph contains at most $3x/2$ edges, the \emph{cluster graph} $H$, in which each cluster of $\mathcal{C}_2$ is contracted to a single \emph{cluster node}, and two such cluster nodes have an edge if any of their nodes are neighbors in $G$,  has at most $3x/2$ edges. As any edge in a cactus graph is contained in at most $2$ cycles, there can be at most $2$ edges between any two clusters of $\mathcal{C}_2$ in $G$, and thus any edge of $H$ corresponds to at most two inter-cluster edges in $G$ (but any inter-cluster edge in $G$ has to correspond to at least one edge in $H$).  Hence, the $3x/2$ edges of $H$ imply a bound of $6x/2$ on the inter-cluster edges between the clusters of $\mathcal{C}_2$.
To get $\mathcal{C}_2$, add one of the endpoints of each inter-cluster edge to the separator set $S$, and see that this decomposition of $G[\alpha\oplus\beta]$, combined with the decomposition $\mathcal{C}_1$ restricted to $G[V\setminus(\alpha\oplus\beta)]$, gives the desired result (see \Cref{lem:cactus_decomp} for details of this step).

\textbf{$(d,x)$-separator decomposition $\mathcal{C}_3$ with $d = O(1/\eps)$ and $x = \eps |\alpha \oplus \beta|$, $k=2\eps M$:}
We merge clusters using \cref{lem:compressingClustering} again and form a clustering with a small number of clusters, while increasing the diameter of each cluster by at most $O(1/\eps)$.
Finally, note that the separator set $S$ of $\mathcal{C}_2$ also separates the computed $(d+O(1/\eps))$-diameter clustering, and gives the desired decomposition (see \Cref{lem:newSeparatorDecomp} for details).
\end{proof}

\subsection{Distributed Computation of Greedy Schedules}\label{sec:mainarboverview}

We conclude with a brief description of the ideas behind \cref{thm:arbOvershooting}. We want to distributively simulate the following greedy algorithm from \cref{thm:arboricitybounds}: process $\alpha$-nodes in an increasing order of degree (in $G[\alpha\oplus\beta]$), and for each of them, say $v$, add all $\beta$-neighbors of $v$ to the vertex cover, then remove $v$. A na\"ive idea is, for rounds $i=1$ to $\Delta$ (max. degree), add  the $\beta$-neighbors of all $\alpha$-nodes of degree $i$ to the vertex cover (in parallel), then remove those $\alpha$-nodes. This does not work, e.g., if all nodes have the same degree. 

To fix this, we need to process $\alpha$-nodes in small groups (of size $\eps \max\{|\alpha|,|\beta|\}$), but still in the degree order. This is achieved by computing an $O(1/\eps)$-diameter clustering, $C_1,\dots,C_k$, with \cref{lem:compressingClustering}, and in each cluster $C$, letting an $\eps$-fraction  of $\alpha$-nodes of given degree $i$, $A_{i,j}^C\subseteq C$, be processed (for $j=1,\dots,O(1/\eps)$ iterations), which effectively means an $\eps$ fraction of \emph{all} degree $i$ nodes are processed in parallel (namely, $\bigcup_{\text{clusters } C} A_{i,j}^C$). Overall, we get an $O(\Delta/\eps)$-batch schedule with the desired approximation (essentially, using the proof of \cref{thm:arboricitybounds}). To reduce the schedule length to $O(\lambda/\eps^2)$, we use the fact that there is only a small fraction of nodes with degree much larger than the arboricity, and those nodes can be processed in a single batch. More formally, we first add to the vertex cover the set $\beta'$ of all $\beta$-vertices of degree larger than $\lambda/\eps$ (they form an $\approx\eps$ fraction of all), then using the bounded degree of vertices in $\beta\setminus \beta'$, we construct a $O(\lambda/\eps^2)$-batch   $(\beta\setminus\beta')$-to-$\alpha$ schedule using the algorithm described above, and revert it  (cf. \cref{obs:reverseschedule}), to finally obtain an $\alpha$-to-$\beta$ schedule. See \cref{ssec:overshooting} for details.

\appendix

\clearpage

\section{Roadmap of Appendix} 
The appendix is structured as follows. 
\begin{itemize}
	\item \cref{app:sequential}: We give full proofs of existential results on various graph classes. 
	
	In \Cref{app:thmsequential} we prove \Cref{thm:sequential} which states the existence of good reconfiguration schedules for cactus graphs and planar graphs. 
	
	In \Cref{ssec:degbased} we show that $(1+\eps)$-type approximations \'a la \Cref{thm:variousClassesBatch} do not extend to bounded arboricity graphs, that is, we prove \Cref{thm:arboricitybounds}. 
	\item \cref{sec:distributed}: We give complementary information for \cref{sec:tecDist}. 
	\item \cref{sec:SLOCAL}: We extend the results of \cref{sec:tecDist} and show how to compute a small separator decomposition on general graphs in $\poly\log n$ rounds. This result can, e.g., be used to compute an $O(1/\eps)$-batch $(1+\eps,O(1/\eps))$-approximation schedules on planar graphs (\cref{thm:distrPlanar}).
	\item \cref{ssec:overshooting}: We show that the quality of our existential approximate schedules for bounded arboricity graphs  can be computed efficiently by a distributed algorithm (\Cref{thm:arbOvershooting}); this result does not use small separator decompositions. 
	\item \cref{app:distrLowerbounds}: We give lower bounds for the run-time of the distributed computation of a reconfiguration schedule. In particular, we show that our distributed algorithms that compute reconfiguration schedules on cactus graphs are asymptotically tight,  and we show that there are  graph classes on which no efficient distributed algorithm can obtain the best (or almost best) existing schedule. 
\end{itemize}

\section{Sequential Reconfiguration Schedules for Various Graph Classes}\label{app:sequential}

In this section, we prove \Cref{thm:sequential,thm:arboricitybounds}. We begin with the latter theorem. It is based on repeated decompositions of graphs and recursive computation of schedules, and leverages the presence of small separators in planar and cactus graphs. The second one studies particular schedules that process the nodes based on the degree sequence.

\subsection{Separable Graphs (Proof of Theorem \Cref{thm:sequential})}
\label{app:thmsequential}

We construct reconfiguration schedules for classes with small separators with an additive approximation that depends on the size of the separators. The same results also follow from the bound of~\cite{deBergJM16} in terms of the \emph{pathwidth} (via a known connection between separators and pathwidth, see e.g.,~\cite[Theorems 20,21]{treewidth}). For the particular class of cactus graphs tighter results can be obtained: such graphs have  $(1,2)$-approximation schedules~\cite{ItoNZ16}. We give a simple and unified proof for both results.

\seqSched*

In the discussion below, we will only consider \emph{monotone} sequential schedules, therefore, we assume that we are given a  graph $G$ with vertex covers $\alpha$ and $\beta$, $\alpha\cap\beta=\emptyset$. The latter implies that every edge $e$ in $G$ has one $\alpha$-vertex and one $\beta$-vertex.

We will consider hereditary graph classes that have small separators. Recall that a graph class is \emph{hereditary} if it is closed under taking induced subgraphs.

Let $f:\mathbb{N}\rightarrow \mathbb{R}_+$ be a positive non-decreasing function. A graph class $\mathcal{F}$ is \emph{$f$-separable} if there is a constant $n_0>0$ such that for every $n>n_0$ and an $n$-node graph $G=(V,E)\in \mathcal{F}$, there is a subset $V'\subseteq V$ of vertices of size $|V'|\le f(n)$ such that $G-V'$ consists of two subgraphs $G_1=(V_1,E_1)$ and $G_2=(V_2,E_2)$, and there is no edge in $G$ between $V_1$ and $V_2$, and $\max(|V_1|,|V_2|)\le 2n/3$. We call $V'$ a \emph{separator}.

Our proofs follow by a inductive application of \cref{lem:discon}.

First, we give a simple generalization of~\cite[Thm. 2]{ItoNZ16}. 
\begin{theorem}\label{thm:bicon}
Let $\mathcal{G}$ be a hereditary graph class, and $\eta\in [1,2]$ and $k\ge 1$ be reals. Assume that for every 2-connected graph $G\in \mathcal{G}$ and vertex covers $\alpha,\beta$ in $G$, there is an $(\eta,k)$-approximation reconfiguration schedule. Then, such schedules exist for every graph in $\mathcal{G}$ with any pair of vertex covers.
\end{theorem}
\begin{proof}
    We prove the claim by induction on the graph size, the base cases being either graphs with at most 2 vertices, or 2-connected graphs in $\mathcal{G}$. Graphs with at most 2 vertices clearly have $(1,1)$-approximation schedules, while $2$-connected graphs have $(\eta,k)$-approximation schedules, by assumption. Let $G$ be any graph in $\mathcal{G}$ with at least 3 vertices and with a vertex $v$ such that $G-v$ is disconnected (i.e., $G$ is not 2-connected). Let $G_1=(V_1,E_1),G_2=(V_2,E_2)$ be two subgraphs (each with at least one vertex) of $G-v$ with no edge between them. Assume, without loss of generality, that $|V_1\cap\beta|-|V_1\cap \alpha|\le |V_2\cap\beta|-|V_2\cap \alpha|$. Let $V'_1=V_1\cup \{v\}$ and $V'_2=V_2$, if $v\in \beta$, and $V'_1=V_1$ and $V'_2=V_2\cup \{v\}$, otherwise. Clearly, $|V'_1\cap\beta|-|V'_1\cap \alpha|\le |V'_2\cap\beta|-|V'_2\cap \alpha|+1$. Since $|V'_1|,|V'_2|<|V|$, by the inductive assumption, there are $(\eta,k)$-approximation schedules $\mathcal{S}_1$ and $\mathcal{S}_2$ for $G[V'_1],V'_1\cap\alpha,V'_1\cap\beta$, and $G[V'_2],V'_2\cap\alpha,V'_2\cap\beta$, respectively; thus, \Cref{lem:discon} applies, proving the induction, i.e., that there is an $(\eta,k)$-approximation schedule for $G,\alpha,\beta$.
\end{proof}

\subsubsection{Cactus Graphs and Forests (Cactus part of \Cref{thm:sequential})}

 \Cref{thm:bicon} 
 implies that every cactus graph 
admits monotone $(1,2)$-approximation schedules, and every forest admits  $(1,1)$-approximation schedules. The theorem asserts that it suffices to limit ourselves to 2-connected instances in each class. In forests, there are no 2-connected instances with more than 2 vertices, so we are done. 
For cacti, it is well-known that every $2$-connected cactus is a cycle graph \cite{geller_manvel_1969}, and it is easy to find a $(1,2)$-approximation schedule for any cycle: add any $\beta$-vertex, to reduce the question to a path, which by the remark above, has a $(1,1)$-approximation schedule.

\begin{theorem}[Cactus part of \Cref{thm:sequential}]
Let $G$ be a cactus graph with vertex covers $\alpha$ and $\beta$. There is a monotone sequential $(1,2)$-approximation schedule for $\alpha,\beta$.
\end{theorem}

\subsubsection{Planar and  Separable Graphs (Planar part of \Cref{thm:sequential})}

In the proof of the following theorem we use  \Cref{lem:discon}. 
\begin{theorem}\label{thm:separable}
Let $\mathcal{F}$ be a hereditary $f$-separable class, for a non-decreasing function $f$. For every $n\in \mathbb{N}$, let  $T(n)$ be the minimum value such that for every $n$-vertex graph $H\in \mathcal{F}$, there is a $(1,T(n))$-approximation schedule for any two disjoint vertex covers in $H$. 
Then there is a constant $n_0>0$ such that for every $n>n_0$, 
$T(n)\le T(\lfloor 2n/3\rfloor)+f(n)$. In particular, $T(n)\le c+f(n)\log_{3/2} n$ holds for a constant $c>0$.
\end{theorem}
\begin{proof}
Let $G\in \mathcal{F}$ be an $n$-vertex graph with vertex covers $\alpha,\beta$, for which there is no  $(1,T(n)-1)$-approximation schedule. 
They exist by the minimality of $T(n)$. Let $V'$ be a 
separator of size $|V'|\le f(n)$  and $G_1=(V_1,E_1),G_2=(V_2,E_2)$ be the two disconnected subgraphs of $G-V'$ with $\max(|V_1|,|V_2|)\le 2n/3$ ($V'$ exists since $\mathcal{F}$ is a $f$-separating family). Since $\mathcal{F}$ is hereditary, we have $G_1,G_2\in \mathcal{F}$. By minimality, $T$ is a non-decreasing function. 
These observations imply that each of $G_1,G_2$ has a $(1,T(\lfloor 2n/3\rfloor))$-approximation schedule. Let $\mathcal{S}_i$ be the batch sequence of such a schedule for $G_i$, $i=1,2$. By \Cref{lem:discon}, 
either $\mathcal{S}_1,\mathcal{S}_2$ or $\mathcal{S}_2,\mathcal{S}_1$ (depending on the counts of $\alpha$ and $\beta$ vertices) is a $(1,T(\lfloor 2n/3\rfloor))$-approximation schedule for $G_1\cup G_2$ (note that in this case, $V_s=\emptyset$). Assume, without loss of generality, it is the former. It follows then that $\mathcal{S}_1,V'\cap\beta,\mathcal{S}_2,V'\cap\alpha$ is a $(1,T(\lfloor 2n/3\rfloor)+|V'|)$-approximation schedule, which implies that $T(n)\le T(\lfloor 2n/3\rfloor)+f(n)$, since $|V'|\le f(n)$. The claim $T(n)\le c+f(n)\log_{3/2} n$ follows from this recursion, since it ends  in roughly $\log_{3/2} n$ iterations, and $f(n)$ is a non-decreasing function. 
\end{proof}

The most prominent class to which the theorem above applies is the class of planar graphs. The celebrated Planar Separator Theorem~\cite{LiptonT80,DjidjevV97} shows that planar graphs are $2\sqrt{n}$-separable.
\begin{theorem}[Proof of Planar part of \Cref{thm:sequential}]\label{thm:planar}
Let $G$ be a planar graph with vertex covers $\alpha$ and $\beta$. Let  $M=\max\{|\alpha|,|\beta|\}$. There is a monotone sequential $(1,O(\sqrt{M}))$-approximation schedule for $\alpha,\beta$.
\end{theorem}
\begin{proof}
We only consider monotone schedules, so we can assume $\alpha\cap\beta=\emptyset$. The Planar Separator Theorem and \cref{thm:separable} imply that for $n$-vertex planar graphs, the additive approximation satisfies $T(n)\le T(\lfloor 2n/3\rfloor)+2\sqrt{n}$, for $n$ larger than a constant $n_0$. We have, therefore, $T(n)\le O(1) + 2\sqrt{n}\cdot \sum_{i=0}^{\infty}(2/3)^{i/2}=O(1)+11\sqrt{n}$. We apply the scheduling to $G[\alpha\oplus\beta]$, which is planar and has at most $2M$ vertices. We get a $(1,T(2M))=(1,O(\sqrt{M}))$-approximation.
\end{proof}

Separator theorems are also known for a number of other graph classes. Graphs of genus $g$ are $O(g\sqrt{n})$-separable~\cite{genus}, graphs  with a forbidden minor with $h$ vertices are $O(h\sqrt{n})$-separable~\cite{KR10}, graphs with treewidth $T$ are $T$-separable~\cite[Thm. 19]{treewidth}, $k$-nearest neighbor graphs in $\mathbb{R}^2$ are $\sqrt{kn}$-separable~\cite{nearestneighbor}. In the theorem below, we merely select a representative class.

\subsection{Bounded Degree/Arboricity Graphs (Proof of \Cref{thm:arboricitybounds})}\label{ssec:degbased}

In this section, we prove \cref{thm:arboricitybounds}. The proof is presented in two parts. First, we present the algorithm, as well as its batched variant, consisting of 4 batches (\Cref{thm:4batch}). Then, we present the impossibility result (\Cref{thm:arblower}). We restate the theorem. 

\thmarboricity*

Recall, that the arboricity of a graph $G=(V,E)$ is the minimum number of forests that are needed to cover the edge set $E$.  Alternatively, the arboricity equals the maximum density of a subgraph, i.e., $\lambda=\max_{U\subseteq V, |U|\ge 2}\left\lceil\frac{|E(G[U])|}{|U|-1}\right\rceil$~\cite{nashWilliams64}.

Consider a graph $G$ with two vertex covers $\alpha$ and $\beta$.
Here we consider the performance of the following natural \emph{monotone sequential greedy} schedule.

Since we deal with monotone schedules, we focus on a schedule from $A=\alpha\setminus \beta$ to $B=\beta\setminus \alpha$. 

\subparagraph*{Sequential greedy schedule.} The schedule iterates over the vertices in $A$ in an increasing order $v_1,v_2,\dots$ of their degree in $G[A\cup B]$ (breaking ties arbitrarily), and for each vertex $v$, adds all its neighbors $N(v)$ (one by one) to the vertex cover, if they haven't been added before,  and then removes $v$ from the current cover.

\smallskip

We first express the cost of the schedule in terms of the parameters of the graph $G[A\cap B]$, particularly the number of vertices and edge density, then we show that this leads to an approximation in terms of the arboricity of $G$.
We will make use of the following simple fact.

\begin{fact}\label{fact:averges}
For any non-decreasing sequence $0\le a_1\le a_2\le \dots\le a_n$ of non-negative reals and index $1\le k\le n$, $\frac{\sum_{i=1}^k a_i}{k}\le \frac{\sum_{i=1}^n a_i}{n}$.
\end{fact}

\begin{lemma}\label{lem:bddeg1b1}
Let $G=(V,E)$ be a graph with two vertex covers $\alpha$ and $\beta$ such that the graph induced by $D=\alpha\oplus \beta$ has $h>0$ edges.
Then the cost of the greedy schedule is at most  $|\alpha\cap\beta|+|A|+|B|-\frac{|A||B|}{h}+1$, 
where $A=\alpha\setminus \beta$ and $B=\beta\setminus \alpha$.
\end{lemma}

\begin{proof}
Any \emph{monotone} reconfiguration schedule from $\alpha$ to $\beta$ in $G$ (including the greedy one) is equivalent to (has the same batch sequence as) a monotone reconfiguration schedule from $A$ to $B$ in $G[D]$. The two schedules only differ by additive $|X|$ in cost, where $X=\alpha\cap\beta$. Thus, we will consider the greedy schedule as a schedule  $\mathcal{S}=(V_i)_{i\in [|D|]}$ for $A$ and $B$ in $G[D]$, which gives us the corresponding schedule $\mathcal{S}=(V_i\cup X)_{i\in [|D|]}$ for $\alpha$ and $\beta$.  

The cost of the schedule (from $A$ to $B$) is the size of the vertex cover (in $G[D]$) obtained at a 
step
with index $1\le i< |D|$ such that in the next 
step 
$i+1$, a vertex $v_i$ is removed from $V_i$. Note that there can be more than one index which realize the cost of the schedule.  
Fix any such $i$ and let $A_i$ be the vertices added until batch $i$ and $B_i$ the vertices removed until batch $i$ (both times inclusive). 
Note that $N(\{v_i\}\cup A_i)=B_i$. 
Then, using \cref{fact:averges}, we have
\[
\frac{|B_i|}{|A_i|+1}\leq\frac{\sum_{v\in A_i\cup\{v_i\}} d_{G[D]}(v)}{|A_i\cup\{v_i\}|}\le \frac{\sum_{v\in A} d_{G[D]}(v)}{|A|}=\frac{h}{|A|}\ ,
\]
implying that $|A_i|\ge |B_i||A|/h-1$. 
The cost of the schedule just after batch $i$ is bounded by
\[
|V_i|=|A|-|A_i|+|B_i|\le |A|+|B_i|\left(1-\frac{|A|}{h}\right)+1\le |A|+|B|-\frac{|A||B|}{h}+1\ .
\]
We obtain the claimed cost by adding the size of $\alpha\cap \beta$ to this bound.
\end{proof}

\cref{lem:bddeg1b1} combined with the bound on the number of edges in a graph with arboricity $\lambda$ yields the following theorem, which gives the positive half of \cref{thm:arboricitybounds}.
\begin{theorem}[Greedy Schedule] \label{thm:seqmindeg}
Let $G=(V,E)$ be a graph with two vertex covers $\alpha$ and $\beta$ such that the graph induced by $D=\alpha\oplus \beta$ has arboricity $\lambda$. 
The greedy schedule is a $(2-\frac{1}{2\lambda},1)$-approximation.
\end{theorem}
\begin{proof}
Let $|A|,|B|$ and $h$ as in \cref{lem:bddeg1b1} and let $M' = \max\{|A|,|B|\}$, $m'=\min\{|A|,|B|\}$ and $X = \alpha\cap\beta$. By \cref{lem:bddeg1b1}, the schedule has cost $|X|+M'+m'-\frac{m'M'}{h}+1$.  
Since $G[D]$ is a graph with $m' + M'$ nodes, $h$ edges, and arboricity $\lambda$, we have that $\lambda \ge \frac{h}{m'+M'-1}$, hence $h\le \lambda(m'+M'-1)\le 2\lambda M'$. 
Thus, recalling that $M=M'+|X|$, the cost of the schedule can be bounded as
\begin{align*}
|X|+M'+m'-\frac{m'M'}{h}+1 &\le |X|+M'+m'-\frac{m'}{2\lambda}+1 \\
&\le |X|+M'+m'\left(1-\frac{1}{2\lambda}\right)+1\\&\le M\left(1-\frac{1}{2\lambda}\right)+1\ ,
\end{align*}
which gives us the claimed approximation.
\end{proof}

Using the schedule provided by \cref{thm:seqmindeg} together with the compression technique from \cref{thm::sched_compres}, we get a $2\lceil 1/\eps\rceil$-batch schedule that is a $(2-\frac{1}{2\lambda}+\eps, 2)$-approximation for arboricity-$\lambda$ graphs, for arbitrarily small $\eps\in (0,1)$. This, however, is non-trivial only for $\eps<1/(2\lambda)$. Here we show, using a more direct approach, that  approximation similar to \cref{thm:seqmindeg} can be achieved with only 4 batches. 
Let $G=(V,E)$ be a graph with two vertex covers $\alpha$ and $\beta$. Again, since we are dealing with only monotone schedules, we will focus on $A=\alpha\setminus\beta$ and $B=\beta\setminus\alpha$.

\subparagraph*{4-batch schedule.} Let $\xi=\left\lfloor\frac{|A||B|}{h+|A|}\right\rfloor$.
Let $X=\alpha\cap \beta$, $A_1$ be the set of the $\xi$ lowest degree (in $G[A\cup B]$) vertices in $A$, and let $B_1=N_{G}(A_1)\cap B$ be the set of neighbors of nodes in $A_1$. The schedule consists of the following four vertex covers: $V_0=X\cup A=\alpha$, $V_1=X\cup A\cup B_1$,
 $V_2=X\cup (A\setminus A_1) \cup B_1$, $V_3=X\cup (A\setminus A_1) \cup B$, $V_4=X\cup B=\beta$.
 
 Again, we first express the cost of this 4-batch schedule in terms of the edge density and other parameters of $G[A\cup B]$, then it will translate into an approximation in terms of arboricity.

 \begin{lemma}\label{lem:4batch}
Let $G=(V,E)$ be a graph with two vertex covers $\alpha$ and $\beta$ such that the graph induced by $D=\alpha\oplus \beta$ has $h>0$ edges. 
The cost of the 4-batch schedule is at most  $|\alpha\cap \beta|+|A|+|B|-\frac{|A||B|}{h+|A|}+1$, 
where $A=\alpha\setminus \beta$ and $B=\beta\setminus \alpha$.
\end{lemma}
 
\begin{proof} 
Since the schedule is monotone, and we never change the status of the nodes in $X=\alpha\cap \beta$, let us assume first that $X=\emptyset$. We analyze the cost of the schedule under this assumption, then add $|X|$ to it, in the case when $X$ is non-empty. 

Now, the 4-batch schedule simplifies to:
 $V_0=A$, $V_1=A\cup B_1$, 
 $V_2=(A\setminus A_1) \cup B_1$, $V_3=(A\setminus A_1) \cup B$, $V_4=B$,
where $A_1$ is the set of $\xi$ lowest degree vertices in $A$ and $B_1=N(A_1)\cap B$. We will show that in order to get the claimed cost, we need to choose $\xi$ as in the algorithm presented above.

The cost of the schedule is the maximum among the size of $V_1$, which is of size $|A|+|B_1|$, and of $V_3$, which is of size $|A|+|B|-\xi$ (note that $V_2 \subseteq V_1$ and thus not included). 
Note that $|B_1|\le \sum_{v\in A_1}d_{G[D]}(v)$, so using \cref{fact:averges}, we have
\[
\frac{|B_1|}{|A_1|}\le \frac{\sum_{v\in A_1}d_{G[D]}(v)}{|A_1|}\le \frac{\sum_{v\in A}d_{G[D]}(v)}{|A|}=\frac{h}{|A|}\ .
\]
Thus, $|B_1|\le |A_1|h/|A|=\xi h /|A|$, hence the cost of the schedule is at most the larger of $|A|+\xi h/|A|$ and $|A|+|B|-\xi$. Let us choose $\xi$ so that those expressions are equal: $|B|-\xi=h\xi/|A|$, so $\xi=|A||B|/(h+|A|)$. This gives us cost $|A|+|B|-\frac{|A||B|}{h+|A|}$. Recall, however, that $\xi$ has to be an integer, so we round it \emph{down}, i.e., we let $\xi=\floor*{|A||B|/(h+|A|)}$, which gives us cost $|A|+|B|-\frac{|A||B|}{h+|A|}+1$ (note that the expression $|A|+\xi h/|A|$ can only be improved by this rounding).
To complete the proof, recall that in order to bound the cost of schedule in the case $X\neq \emptyset$, we only need to add $|X|$ to the cost above. 
 \end{proof}

\begin{theorem}\label{thm:4batch}
Let $G=(V,E)$ be a graph with two vertex covers $\alpha$ and $\beta$ such that the graph induced by $D=\alpha\oplus \beta$ has arboricity $\lambda$. 
The 4-batch schedule is a $(2-\frac{1}{2\lambda+1},1)$-approximation.
\end{theorem}

\begin{proof} 
The proof follows along the lines of the proof of \cref{thm:seqmindeg}, with the only difference that in \cref{lem:4batch}, we have the term $\frac{|A||B|}{h+|A|}$ instead of the $\frac{|A||B|}{h}$ term of \cref{lem:4batch}. The additional $|A|$ term is that gives rise to $+1$ in $\frac{1}{2\lambda+1}$ term in the approximation.
\end{proof}

\subparagraph*{How well can we schedule bounded degree/arboricity graphs?}
We prove below that for every $d\ge 4$, there exist $d$-regular bipartite graphs $G=(A\cup B, E)$ of arbitrarily large size, such that any vertex cover reconfiguration schedule from $A$ to $B$ is no better than a $(2-\Theta(1)/{\sqrt{d}}, 0)$-approximation. It follows, in particular, that for every $\lambda\ge 4$, there are graphs of arboricity at most $\lambda$ where every reconfiguration schedule  is at best a $(2-\Theta(1)/\sqrt{\lambda}, 0)$-approximation.
This proves the negative part of \cref{thm:arboricitybounds}.

\begin{theorem}[Arboricity Lower bound]\label{thm:arblower}
For every $4 \leq d < N/2$, there is a bipartite $d$-regular graph $G=(A\cup B, E)$ with $|A|=|B|=n\in [N, 2N]$, such that any vertex cover reconfiguration schedule from $A$ to $B$ is no better than a $\left(2-\frac{4}{\sqrt{d+1}+1}, 0\right)$-approximation.
\end{theorem} 

We will use several known results on Ramanujan graphs. A $d$-regular bipartite graph $G=(A\cup B, E)$ is \emph{Ramanujan} if the largest non-trivial eigenvalue of its adjacency matrix is in the interval $[-2\sqrt{d-1}, 2\sqrt{d-1}]$ (such a matrix always has two \emph{trivial} eigenvalues, $d$ and $-d$). 
\begin{theorem}[{\cite[Theorem 5.5]{MarcusSS13}}]\label{thm:ramexists} For every $d\ge 3$ and $N>2d$, there exist $d$-regular bipartite Ramanujan graphs with $n\in [N,2N]$ nodes.
\end{theorem}

The \emph{incidence matrix} $M$ of a bipartite graph $G=(A\cup B, E)$ is a $|A|\times |B|$ binary matrix, where the $(i,j)$ entry is $1$ if and only if the $i$th vertex in $A$ is adjacent to the $j$th vertex in $B$. Let $\delta_1\ge\delta_2\ge\dots\ge\delta_{|A|}$ be the eigenvalues of the symmetric matrix $MM^T$, where $M^T$ is the transpose of $M$.
\begin{theorem}[{\cite[Theorem 2.1]{Tanner84}}]\label{thm:tan}
Let $G=(A\cup B, E)$ be a $d$-regular bipartite graph, such that $\delta_1>\delta_2$. Then for every subset $S\subseteq A$ of vertices,  it holds that 
\[
|N(S)|\ge \frac{d^2|S|}{\frac{|S|}{|A|}(d^2-\delta_2)+\delta_2}\ .
\]
\end{theorem}

Note that the right hand side of the inequality can be transformed into the following form that is more convenient for us:
\[
|N(S)|\ge \frac{|A|}{1+\frac{|A|}{|S|}\cdot \frac{\delta_2}{d^2}}=|A|\left(1-\left(1+\frac{|S|}{|A|}\cdot \frac{d^2}{\delta_2}\right)^{-1}\right)\ .
\]
It is easy to verify that for every bipartite graph $G$ with adjacency matrix $X$ and incidence matrix $M$, every eigenvalue of $MM^T$ is also an eigenvalue of $XX^T$, and hence is a \emph{square of an eigenvalue} of $X$ (compare \cite[Theorem 4.15]{HooryL06} with the theorem above in this context). Also note that $\delta=d^2$ is an eigenvalue of $MM^T$, corresponding to the all-ones eigenvector. 
Therefore, if the graph is $d$-regular, $d\ge 4$, and Ramanujan, we have that $d^2=\delta_1>4(d-1)\ge \delta_2$, and \cref{thm:tan} applies, giving us the following. 
\begin{corollary}\label{c:ramanujan}
Let $G=(A\cup B, E)$ be a $d$-regular bipartite Ramanujan graph, for $d\ge 4$. Then for every subset $S\subseteq A$ of vertices it holds that 
\[
|N(S)|\ge |A|\left(1-\left(1+\frac{|S|}{|A|}\cdot \frac{d^2}{4(d-1)}\right)^{-1}\right)\ .
\]  
\end{corollary}

\begin{proof}[Proof of \cref{thm:arblower}]
Let $G=(A\cup B, E)$ be a $d$-regular bipartite Ramanujan  graph with $|A|=|B|=n$ provided by \cref{thm:ramexists}. 
Consider a reconfiguration schedule $\mathcal{S}=V_0,V_1,\dots,V_{\ell}$ from  $A$ to $B$. Let $V_i$ be the first vertex cover for which $|A\setminus V_i|\ge \eta n$,
where $\eta\in (1/n,1)$ is a parameter we will specify below. Let $C=A\setminus V_i$ and note that $N(C)\subseteq V_{i-1}$. Thus, since $|A\setminus V_{i-1}|< \eta n$, we have $|V_{i-1}|\ge |A|-|A\setminus V_{i-1}| + |N(C)|\ge n-\eta n+|N(C)|$. From \cref{c:ramanujan}, 
we also have that
\[
|N(C)|\ge n\cdot \left(1-\left(1+\frac{\eta d^2}{4(d-1)}\right)^{-1}\right)\ge n\cdot \left(1-\frac{1}{1+\eta d/4}\right)\ ,
\]
which implies that
\[
|V_{i-1}|\ge n\left(2-\eta - \frac{1}{1+\eta d/4}\right)\ .
\]
Optimizing for $\eta$, we see that the best choice is when $\eta$ and $\frac{1}{1+\eta d/4}$ are balanced, that is, $(d/4)\eta^2+\eta=1$. This means that denoting $c=d/4$, we obtain \[
\eta=\frac{-1+\sqrt{4c+1}}{2c}=\frac{(-1+\sqrt{4c+1})(1+\sqrt{4c+1})}{2c(1+\sqrt{4c+1})}=\frac{2}{\sqrt{4c+1}+1}=\frac{2}{\sqrt{d+1}+1}\ .
\]
Thus, the claim follows as the cost of the schedule is at least
\[
|V_{i-1}|\ge n\left(2-2\eta\right)=n\left(2-\frac{4}{\sqrt{d+1}+1}\right)\ ,\qedhere
\]
\end{proof}

\begin{proof}[Proof of \Cref{thm:arboricitybounds}]
	The existence of a schedule is precisely \Cref{thm:4batch}, and the impossibility result is \Cref{thm:arblower}. 
\end{proof}

\section{Distributed Computation of Schedules I (Cactus Graphs)}
\label{sec:distributed}

The objective of this section is to provide the formal proof of \Cref{thm:cactus_core_graphs}. As its  proof only requires statements that we have already presented in \Cref{sec:tecDist} we begin with a proving the theorem. Afterwards we prove the statements whose proofs are omitted in \Cref{sec:tecDist}, namely we prove \cref{gen::lma::decomp_to_sched_computation} and provide additional intuition for small separator decompositions in \Cref{sec:generalFramework}, we prove  \cref{lem:compressingClustering} (cluster merging) in \Cref{ssec:clusterMerging}, and we prove \Cref{lem:cactus_core_decomp} (computing small separator decompositions for cactus graphs) in \Cref{ssec:using}. The central tool of this section is the notion of a small separator decomposition which allows simultaneous reconfiguration of `distant clusters' of the graph. The small separator decomposition is also used in \Cref{sec:SLOCAL} to compute reconfiguration schedules for planar and outerplanar graphs. 

\thmCactusCoreRecon*

\begin{proof}
	Let $D=\alpha\oplus \beta$. By our assumption, $G[D]$ is a cactus graph. 
	Use \cref{lem:cactus_core_decomp} to compute a $(d,x)$-separator decomposition $\mathcal{C}=\{C_1,\dots,C_{k}\}$ of $G$ with ghost node set $Z=V\setminus (\alpha\oplus \beta)$, separator set $S$ and parameters $d=O(1/\eps)$, $x=\eps |D|$ and $k= \max\{1, 2 \eps M\}$ (recall that $M = \max\{|\alpha|,|\beta|\}$).
	
	Now, apply \cref{gen::lma::decomp_to_sched_computation} to the clustering $\mathcal{C}$ to compute a reconfiguration schedule: 
	For each cluster $C\in \mathcal{C}$, $G[D \cap C]$ is a cactus graph, hence, from \cref{thm:variousClassesBatch},  $G[C\cap D]$ (and thus also $G[C]$) has an $O(1/\eps)$-batch $(1+\eps,3)$-approximation schedule. Thus, we can apply \cref{gen::lma::decomp_to_sched_computation} on the clustering $\mathcal{C}$ and obtain an $O(1/\eps)$-batch $(1 + \eps,\eps |D| + k\cdot 3)$-approximation schedule, that is, in total we can upper bound the size by 
	\begin{align*}
		(1+\eps) M+\eps|D|+k\cdot 3\stackrel{|D|\leq 2M}{\leq} (1+\eps) M+\eps \cdot 2M + \max\{1, 2 \eps M\} \cdot 3 \leq (1+9\eps)M+3,
	\end{align*}
	showing that the schedule is a $(1+9\eps,3)$-approximation. Rescaling $\eps$ provides the result.

	For the runtime,  applying \cref{lem:cactus_core_decomp} takes $O(\logstar(n)/\eps)$ rounds.  \cref{gen::lma::decomp_to_sched_computation} takes $O(d)=O(1/\eps)$ rounds. Thus, the algorithm consists of $O(\logstar(n)/\eps)$ rounds in total.
\end{proof}
Due to \cref{obs:eh_free}, \cref{thm:cactus_core_graphs} can be applied if $G$ belongs to one of various graph classes.

\corfastDistributedReconf*

\begin{proof} To apply \cref{thm:cactus_core_graphs}, it suffices to notice, as we did in \cref{thm:variousClassesBatch}, that for each graph class above,  $G[\alpha\oplus\beta]$ is a cactus graph.
\end{proof}

\cref{thm:variousClassesBatch} shows that even-hole-free graphs have a $(\ceil{1/\eps}+1)$-batch $(1+\eps,3)$-approximation schedule, and by \cref{thm::t_lb} this schedule is almost optimal. \cref{cor:variousClassesDist} gets almost the same schedule as in \cref{thm:variousClassesBatch}, with the same approximation factor and a length of the same order in $O(\logstar(n)/\eps$) rounds.

\subsection{Small Separator Decompositions} 
\label{sec:generalFramework}
In this section, we recall the purely graph-theoretic small separator decomposition, give additional intuition and provide missing proofs in using it to quickly compute batch reconfiguration schedules with a distributed algorithm. Algorithms to compute such decompositions are postponed to \Cref{ssec:using} and \Cref{sec:SLOCAL}. 

Let $U\subseteq V$ be a subset of the vertices of a given graph $G=(V,E)$. The \emph{weak diameter} of $U$ is the maximum distance measured in $G$ between any two vertices of $U$. The \emph{strong diameter} (or simply diameter) is the diameter of the graph $G[U]$, that is, it is the maximum distance between any two vertices of $U$ measured in the graph $G[U]$.

\DefsmallSeparator*

The separating property of the set $S$ in a $(d,x)$-separator decomposition can be grasped easily by setting $Z=\emptyset$. In this case, the clusters and the separator set partition all vertices of the graph such that there is no edge between any two clusters, that is, any path between two clusters has to go through $S$. If $Z$ is nonempty, at least one endpoint of any inter-cluster edge has to be in $Z$. 
We emphasize that the clusters in a $(d,x)$-separator decomposition can be disconnected.  If we want to compute a $(d,x)$-separator decomposition in the distributed setting we require that each cluster is equipped with a unique ID and a leader node that does not have to be contained in the cluster. Each node of the cluster needs to know the cluster ID and a path of length at most $d$ to the cluster leader.  In simultaneous and independent work, a similar but different decomposition has been used to compute $(1+\eps)$-approximations of the minimum vertex cover problem in the CONGEST model, deterministically and in $\poly(\eps^{-1}, \log n)$ rounds \cite{FK20}.
The next lemma shows that a small separator decomposition is helpful to compute batch reconfiguration schedules.

\lemseparatorHelps*

The proof of \Cref{gen::lma::decomp_to_sched} appears in \Cref{sec:tecDist}.

\lemdecomptoschedcomputation*

\begin{proof}
We show how to construct the required schedule given in \cref{gen::lma::decomp_to_sched}. Consider its four phases.
Phases 1 and 4 can each be implemented in a single round. To compute the schedule of each cluster (in phases 2 and 3), each node learns its entire cluster in $O(d)$ rounds and locally computes its position in the schedule.
\end{proof}

In the proof of \Cref{gen::lma::decomp_to_sched} we reconfigure clusters of the given small separator decomposition in a specific order. The next remark explains why a na\"ive usage of small separator decompositions that does reconfigure all clusters at the same time can be hurtful for the approximation guarantee. 
\begin{remark} 
	\label{remark:nonNaive}
 Assume, for a given small separator decomposition,  we can promise some good approximation on each cluster $C\in \mathcal{C}$, e.g., an $(\eta,c)$-approximation.  Note that the approximation guarantee is only with respect to the maximum between $C_\alpha = C\cap\alpha$ and $C_\beta = C\cap\beta$. Thus, simultaneously reconfiguring a cluster $C$ with $|C_\alpha| \gg |C_\beta|$ and a cluster $C'$ with $|C'_\alpha| \ll |C'_\beta|$ at the same time might result in a much worse overall approximation. 

For illustration, consider a graph composed of the disjoint union of two stars  $S_1$ and $S_2$ of the same size (i.e., $K_{1,t}$ for some integer $t$) with $\alpha$ set to be the center of $S_1$ and the leaves of $S_2$ and $\beta$ defined as its complement. This graph immediately has a $4$-batch $(1,1)$-approximation, which adds $\beta\cap S_2$, removes $\alpha\cap S_2$, adds $\beta\cap S_1$ and then removes $\alpha \cap S_1$. On the other hand, see that a $(1,1)$-approximation on $S_1$ (resp. $S_2$) must add all of $S_1$ (resp. $S_2$) to the vertex cover before removing a single node from it. Therefore, applying the schedule of the stars simultaneously might result in a trivial  $(2,0)$-approximation algorithm.
\end{remark}

\subsection{Cluster Merging (Proof of \cref{lem:compressingClustering})}
\label{ssec:clusterMerging}
In the following Lemma, we show that given a partition of the graph into connected clusters, we can quickly merge some clusters in the partition to ensure that each cluster either has diameter bounded from below or consists of a whole connected component of the graph. Its proof is inspired by the well-known GHS and GKP algorithms for the distributed computation of a minimum-weight spanning tree \cite{GallagerHS83_MST,GarayKP98_MST}. We use this subroutine in several parts of the paper, in particular, we use it at the end of this section to transform a small separator decomposition of $G[\alpha\oplus\beta]$ with potentially many distinct clusters into a small separator decomposition of all of $G$ that only has few clusters (cf. \cref{lem:newSeparatorDecomp}), which is essential to keep the cost low when applying \cref{gen::lma::decomp_to_sched}.

We emphasize that the clusters in a $d$-diameter clustering need to be connected and that all vertices of the graph are contained in some cluster.

\lemClusterMerging*

\begin{proof}
We describe the algorithm that merges clusters from $\mathcal{C}$ to obtain $\mathcal{C'}$. To this end, assume that initially, $\mathcal{C'}=(C_i)_{i=1}^k$ is identical to $\mathcal{C}$. 
For simplicity of exposition, we consider the \emph{cluster-graph} $G'$ with the set $\mathcal{C}'$ of vertices, where two clusters are \emph{adjacent} if any of their corresponding vertices are adjacent in $G$. We describe the algorithm using $G'$, and explain below how this can be implemented in $G$. 
The algorithm consists of two stages, the first consists of $O(\log (1/\eps))$ iterations, while the second happens only once. 

\subparagraph*{Stage 1.} This stage consists of iterations $t=0,1,2,\dots,\lceil\log_{2}(1/\eps)\rceil$. Next, we describe an iteration $t$.
A cluster $C_i$ is \emph{small} (in iteration $t$) if it has diameter at most $2^t$ and otherwise it is \emph{large}. A small cluster is \emph{isolated} if it is not adjacent to any other small cluster (it could still have edges to large clusters).

We update the clustering $\mathcal{C'}$ by merging several small clusters together. 
To this end, every small cluster $C'_i$ (thought of as a node in $G'$) picks an arbitrary edge $\{v,u\}\in E(G')$ connecting it to another small cluster $C'_j$, if such an edge exists. Let $F$ be the set of selected edges. The edges in $F$ are oriented as follows: if an edge $\{C'_i,C'_j\}$ was picked only by $C'_i$, it becomes an out-edge of $C'_i$, i.e., oriented as $(C'_i,C'_j)$, otherwise (if $C'_i$ and $C'_j$ both picked $\{C'_i,C'_j\}$) it is oriented arbitrarily. This gives us a pseudoforest $\vec F$ (in $G'$), i.e., a directed subgraph with maximum outdegree 1. We compute a maximal independent set $I$ in $\vec F$  with the algorithm from~\cite{Cole_1986}.

We update the clustering $\mathcal{C}'$, as follows: every non-isolated small cluster that is not in $I$ picks an arbitrary adjacent cluster in $I$ and merges with it. That is, given a cluster $C'_i\in I$, let $C'_{i_1},\dots,C'_{i_l}$ be the clusters that merge with $C'_i$. We remove $C'_i$ and $C'_{i_1},\dots,C'_{i_l}$   from $\mathcal{C}'$ and add a new cluster $C'=C'_i\cup C'_{i_1}\cup\dots\cup C'_{i_l}$. This completes the description of one iteration of Stage 1.

\subparagraph*{Stage 2.} Every cluster with diameter smaller than $1/\eps$ picks an arbitrary adjacent cluster, if there is any, and merges with it. 

~\\\noindent This completes the description of the algorithm. 

\medskip

\subparagraph*{Diameter bounds.}
Towards the rest of the analysis, let us first note that after each iteration $t$ of Stage 1, due to the maximality of the independent set $I$, if a small cluster was not merged, then all of its neighbors in $G'$ must have diameter greater than $2^t$. Therefore, after the last iteration ($t=\lceil \log(1/\eps)\rceil$) of Stage 1, every cluster with diameter smaller than $1/\eps$ (i.e., a small cluster) that is not a whole connected component of $G$, is adjacent to a cluster of diameter at least $1/\eps$, and hence is merged with such a cluster in Stage 2. 
Thus, after Stage 2, every cluster is either an entire connected component of $G$ or has diameter at least $1/\eps$. 

Next, let us see that the diameter of a cluster does not increase abruptly within a single step of the algorithm. In iteration $t$ of Stage 1, only clusters with diameter at most $2^t$ participate in a merge, and whenever a cluster merges with a set of other clusters they together induce a subgraph of $G'$ of diameter at most $2$ (which implies that the longest chain of such clusters is of length 3). Thus, after iteration $t$, every new cluster has diameter at most $3\cdot 2^{t}+2$, that is, $2^t$ for every cluster in any chain of merged clusters and the additive $2$ for the edges connecting the clusters in the chain. Therefore, after Stage 1, every cluster has diameter at most $\max\curlybrackets*{d,3\cdot \lceil 1/\eps\rceil+2} = d + O(1/\eps)$. In Stage 2, every cluster of diameter smaller than $1/\eps$ that is not a separate component of $G$ merges into another cluster, hence, after this stage, the maximum diameter of a cluster is increased by at most $2 + 2/\eps$. This yields that the maximum diameter of a cluster is still at most $d + O(1/\eps)$.

\subparagraph*{Runtime and implementation in $G$.} 
First, we can use a preprocessing stage (before Stage 1) of $O(d)$ rounds, to let each cluster elect a leader and compute its diameter. Note that in iteration $t$ of Stage 1, 
all clusters that participated have diameter bounded by $2^t$ and thus one round of an algorithm on the cluster-graph in iteration $t$ can be simulated in $O(2^t)$ rounds in $G$. Computing a maximal independent with the algorithm from \cite{Cole_1986} uses $O(\logstar \eta)$ \emph{cluster-graph rounds} where $\eta$ is the maximum ID of a node in $\vec F$ (here, $\eta=poly(|V|)$, as each $C_i$ can assume the ID of its leader).
Therefore, the first stage has runtime
\begin{align*}
    \sum_{t=1}^{\ceil*{\log(1/\eps)}}O(\logstar(n) \cdot 2^t)  = O\parenths*{\frac{\logstar n }{\eps}}.
\end{align*}
Similarly, Stage 2 takes $O(1/\eps)$ rounds in $G$.
\end{proof}

The value of $\eps$ in \cref{lem:compressingClustering} can either be a constant or chosen subconstant and can depend on $n$ without violating the correctness of the claim.
Next, we show that any small separator decomposition of $G[\alpha\oplus \beta]$ can be transferred into a small separator decomposition of $G$ with few clusters ($O(\eps M)$ with $M=\max\{|\alpha|,|\beta|\}$), where the ghost node set can be chosen in order to apply \cref{gen::lma::decomp_to_sched_computation}. We begin with a simple observation about vertex covers.

Next, we use \Cref{lem:compressingClustering} to reduce the number of clusters in a small separator decomposition without increasing the size of the separator set and with only mildly increasing the cluster diameter.
\begin{restatable}{lemma}{newSeparatorDecomp}
	\label{lem:newSeparatorDecomp}
	For any $\eps>0$, $d\geq 1$ there is a deterministic \LOCAL algorithm that, given a graph $G=(V,E)$ with two vertex covers $\alpha$ and $\beta$ and a $(d,x)$-separator decomposition of $G[\alpha\oplus \beta]$ with an empty ghost node set, computes a $(d + O(1/\eps),x)$-separator decomposition of $G$ with a ghost node set $Z=V\setminus (\alpha \oplus \beta)$ and $\max\{1,2\eps M\}$ clusters in $O(\logstar(n)/\eps+d)$ rounds.
\end{restatable}

We emphasize that \cref{lem:newSeparatorDecomp} can deal with disconnected input clusters as long as their weak diameter is small which is along the definition of a $(d,x)$-separator decomposition. As such, the output clusters of \cref{lem:newSeparatorDecomp} are not necessarily connected (i.e. each cluster might be formed of several disconnected components).

\begin{proof} 
Let $D=\alpha\oplus\beta$, $X=\alpha\cap\beta$ and let $\mathcal{C}=\{C_1,\dots,C_l\}$  be the $(d,x)$-separator decomposition of $G[D]$ with separator set $S\subseteq D$ where each connected component of a cluster of the input separator decomposition is its own cluster.  Note that $D\cup X=\alpha\cup \beta$. We say that $\mathcal{C}$ is the \textbf{first clustering}.

The \textbf{second clustering} is formed by the clusters in $\mathcal{C}$ and a new cluster for every vertex that is not contained in $C_1\cup \dots \cup C_l$. Each such new cluster consists of a single node.  
We obtain a $d$-diameter clustering of $G$ (with the same $d$ as in the first clustering). 

\subparagraph*{Third clustering.} Apply  \cref{lem:compressingClustering} to the second clustering and obtain a $(d + O(1/\eps))$-diameter clustering $\mathcal{C'}$ of $G$. If $\mathcal{C'}$ consists of a single component, then clearly the number of components is at most $\max\{1,2\eps  M\}$. Otherwise, $\mathcal{C'}$ consists of at least $2$ components, and by  \cref{lem:compressingClustering}, since $G$ is connected, each cluster has minimum diameter $1/\eps$. Since every vertex cover of a path contains at least half of its vertices, % By \cref{lem:diameterVC}, 
the diameter $1/\eps$ implies that every vertex cover of $G$ must have at least $1/2\eps$ nodes in $C'_i$ and so, we conclude that there are at most $2\eps M$ clusters.

\subparagraph*{Fourth clustering.} Now, notice that each cluster $C'$ consists of $S$-nodes, clusters $C_{i}$ (from the clustering $\mathcal{C}$), and $Z$-nodes, where we defined $Z=V\setminus D$. We remove from each $C'$ its $S$-nodes, and claim that this yields a $(d + O(1/\eps),x)$-separator decomposition with a ghost node set $Z$ and at most $\ceil*{2\eps M}$ clusters.
\begin{enumerate}
    \item Before removing the $S$-nodes, each cluster $C'$ had a strong diameter of $O(1/\eps)$. After the removal, it might become disconnected, but must still have a weak diameter of $O(1/\eps)$.

    \item By the definition of $S$ (in the first clustering), we have $|S| \leq x$.
    
    \item From the definition of a $d$-diameter clustering, before removing the $S$-nodes, $\mathcal{C}'$ was a partition of $V$. Hence, $S$ together with the new clusters again form a partition of $V$. 
		
    \item $S$ separates: Note that every component $C_{i}$ of $\mathcal{C}$ is contained in a distinct cluster $C'$ of $\mathcal{C}'$, and for every $C'$, $C' \setminus Z$ is composed only of such $C_{i}$ components. Therefore, together with the fact that $S$ separates $\{C_{i}\}_{i}$, this shows that $S$ separates $\{C'_j \setminus Z\}_j$.\qedhere
\end{enumerate}
\end{proof}

\subsection{Computing Small Separator Decompositions for Cactus Graphs}
\label{ssec:using}
In this section we show that we can efficiently compute small separator decompositions, if $G[\alpha\oplus\beta]$ is a cactus graph. In fact, we prove the following restated lemma. 

\lemCactusCoreDecomp*

We begin with the following observation on cactus graphs that is crucial for bounding the number of separators in our algorithm. Recall that a graph $H$ is a minor of a graph $G$ if $H$ can be obtained from $G$ by removing edges, vertices, or contracting edges (merging their endpoints).
\begin{observation}[Cactus graphs]\label{obs:cactusMatching}\label{obs:cactusMinor}
    1. The family of cactus graphs is closed under taking minors.
    2. Every connected cactus graph on $n$ vertices contains at most $3n/2$ edges.
\end{observation}
\begin{proof} For the first claim, observe that if after removing vertices or edges, or contracting edges, an edge belongs to at least two cycles (i.e., the obtained graph is not cactus), then it will belong to at least two cycles after reverting the above operations too.
To bound the number of edges, observe that all cactus graphs are simple, and since each edge belongs to at most 1 cycle, one can remove an edge from each cycle to obtain a tree, hence the number of edges is $m=n-1+c$, where $c$ is the number of cycles. Also, each cycle contains at least 3 edges, and every edge belongs to at most one cycle, hence, $c\le m/3$. Thus, $m-n+1\le m/3$, which implies that $m< 3n /2$.
\end{proof}

In the proof of the following statement we use the cluster merging procedures developed in \Cref{ssec:clusterMerging} to compute a decomposition of the graph with small diameter clusters; then we add one vertex of each inter-cluster edge to the separator set and \Cref{obs:cactusMinor} is crucial to show that the separator set is small. 

\begin{restatable}[Small separator set, many clusters]{lemma}{lemCactusDecomp}\label{lem:cactus_decomp}
	For any $\eps > 0$ and for any $n$-node connected cactus graph $G=(V,E)$ with a ghost set $Z = \emptyset$, there exists a $(d,x)$-separator decomposition with $d = O(1/\eps)$ and $x = \eps n$, and with the additional property that the strong diameter of each cluster of the decomposition is at most $d$. Furthermore, such a decomposition can be found in $O\parenths*{\logstar(n)/\eps}$ rounds in \LOCAL.
\end{restatable}

\begin{proof}
Apply \cref{lem:compressingClustering} on the 0-diameter initial clustering of $G$, with every cluster consisting of a single node $v \in V$.  
This  yields a clustering $\mathcal{C}'= \{C'_1,\hdots,C'_k$\} with each cluster having a strong diameter $O(1/\eps)$, in $O(\logstar(n)/\eps)$ rounds. 
If $k = 1$, we let $S = \emptyset$ to get the desired decomposition. On the other hand, if $k \geq 2$, recall that \cref{lem:compressingClustering} gives in this case clusters of diameter at least $1/\eps$ (as $G$ is connected), and so, the total number of clusters is bounded by $k\leq \eps n$. 

It follows from the definition of a $d$-diameter clustering and the assumption that $G$ is connected that each cluster induces a connected subgraph of $G$. Let us now bound the number of inter-cluster edges $I$ (i.e., edges with endpoints in different clusters) for $k\geq 2$. First, observe that there can be at most 2 inter-cluster edges connecting any two given clusters, since otherwise, using the fact that each cluster induces a connected subgraph, we can find two cycles that share an edge, which contradicts the assumption that $G$ is a cactus graph.
Let us form the \emph{cluster-graph} $H=(\mathcal{C}', E')$, where two vertices $C'_i,C'_j\in\mathcal{C}'$ are adjacent if and only if there is an inter-cluster edge in $G$ with one endpoint in $C'_i$ and the other in $C'_j$. The observation above implies that an edge in $E'$ can correspond to at most 2 inter-cluster edges, i.e., $|E'|\ge |I|/2$. On the other hand, since each cluster $C'_i$ induces a connected subgraph of $G$, observe that $H$ is a simple minor of $G$. Recall from  \cref{obs:cactusMinor} that cactus graphs are closed under taking minors, hence $H$ is a cactus graph, which implies by  \cref{obs:cactusMatching} that the number of its edges can be bounded by
\begin{align*}
    |\mathcal{C'}| - 1 + |M(H)| \leq \frac{3}{2}|\mathcal{C'}|= \frac{3}{2}k,
\end{align*}
where $M(H)$ is a maximum matching of $H$. Putting all together, thus the number of inter-cluster edges is at most $|I|\leq 2|E'|\leq 3k$.

To construct $S$, for each inter-cluster edge, arbitrarily add one of its two endpoints to $S$ and remove that node from its cluster, while the other endpoint remains in its cluster. These removals might disconnect clusters in $\mathcal{C'}$, creating a new clustering $\mathcal{C}$ with (possibly) more clusters; however, this process does not introduce new edges between clusters. Thus $S$ separates the clusters $\mathcal{C}$ and we get that $|S|$ is exactly the number of inter-cluster edges, hence, at most $3k\leq 3\eps n$. The clusters of $\mathcal{C}$ along with $S$ provide the desired decomposition. Since $\eps$ was arbitrary, we get the claim of the lemma.
\end{proof}

To obtain \Cref{lem:cactus_core_decomp} from the small separator decomposition provided by \Cref{lem:cactus_decomp} we use \Cref{lem:newSeparatorDecomp} to reduce the number of clusters. 

\begin{proof}[Proof of \Cref{lem:cactus_core_decomp}]
Let $D=\alpha\oplus \beta$. By our assumption, $G[D]$ is a cactus graph with components $G_1,\hdots,G_k$. 
We first create a small separator decomposition of $G[D]$, which we later extend to a decomposition of the whole graph.

Apply the cactus graph decomposition in \cref{lem:cactus_decomp} to each component $G_i$ to obtain a clustering $\{C_{ij}\}_{j=1,\hdots,k_i}$ and a separator set $S_i$, for each such $G_i$. We use those to form a clustering $\mathcal{C} = \{C_{ij}\}_{j=1,\hdots,k_i,i=1,\hdots, k}$ and a separator set $S = \bigcup_{i=1}^{k}S_i$ of $G[D]$. Note that by \cref{lem:cactus_decomp}, each $C_{ij}$ is a connected cactus graph with diameter $O(1/\eps)$ and the size of $S$ is bounded by $\sum_{i=1}^{k}\eps |G_i| = \eps |D|$. The obtained clustering is a $(d,x)$-separator decomposition of $G[D]$ with an empty ghost node set,  separator set $S$, and parameters $d=O(1/\eps)$ and $x=\eps |D|$.
Now, apply \cref{lem:newSeparatorDecomp} to obtain a $(d,x)$-separator decomposition $\mathcal{C'}=\{C_1',\dots,C_{k'}'\}$ of $G$ with ghost node set $Z=V\setminus (\alpha\oplus \beta)$, separator set $S$ and parameters $d=O(1/\eps)$, $x=\eps |D|$ and $k'= \max\{1, 2 \eps M\}$ (recall that $M = \max\{|\alpha|,|\beta|\}$).

For the runtime, see that applying \cref{lem:cactus_decomp} to each component takes $O(\logstar(n)/\eps)$ rounds and \cref{lem:newSeparatorDecomp} takes $O(\logstar(n)/\eps+d)=O(\logstar(n)/\eps)$ rounds. Thus, the algorithm consists of $O(\logstar(n)/\eps)$ rounds in total.
\end{proof}

\section{Distributed Computation of Schedules II (Planar Graphs)}
\label{sec:SLOCAL}
The main objective of this section is to prove the following result.
\begin{theorem}[Reconfiguration on Planar Graphs]
	\label{thm:distrPlanar}
	Let $\eps \in (0,1)$. There exists a deterministic $O(\poly\log (n) /\eps^2)$ round \LOCAL algorithm to  compute a  $O(1/\eps)$-batch $(1+\eps,O(1/\eps))$-approximation vertex cover reconfiguration schedules on planar graphs.
\end{theorem}

To prove \cref{thm:distrPlanar}, we use \cref{thm:genGraphSepDecomp}, which shows that one can compute efficiently a small separator decomposition on any graph.
\begin{theorem}[Distributed Small Separator Decomposition] \label{thm:genGraphSepDecomp}
	For any $\eps>0$ there is a deterministic \LOCAL algorithm that for any connected $n$-node graph $G=(V,E)$ with vertex covers $\alpha$ and $\beta$ computes a $(d,x)$-separator decomposition with $d=O(\log n/\eps)$ and $x=\eps |\alpha \oplus \beta|$ and $k = \max\{1,2\eps M\}$ clusters ($M = \max\{|\alpha|,|\beta|\}$)  in $O(\poly\log (n) /\eps)$ rounds.
\end{theorem}
 We postpone the formal proof of \cref{thm:genGraphSepDecomp} and prove \cref{thm:distrPlanar} assuming its correctness. Note that a cluster (provided by \Cref{thm:genGraphSepDecomp}) in a planar 
 graph is also planar;  
 thus, we can use \Cref{thm:variousClassesBatch} to obtain reconfiguration schedules for each cluster, as needed by \Cref{gen::lma::decomp_to_sched_computation}.

\begin{proof}[Proof of \cref{thm:distrPlanar}]
We apply \cref{thm:genGraphSepDecomp} to compute a $(d,x)$-separator decomposition with parameters $d=O(\log n/\eps^2)$, $x=\eps^2 |\alpha\oplus\beta|$ and $k= \max\{1,2 \eps^2 M\}$ in $O(\poly\log (n) /\eps^2)$ rounds. Every cluster is planar as well, and so, by \cref{thm:variousClassesBatch} it admits an $O(1/\eps)$-batch $(1+\eps,O(1/\eps))$-approximation schedule. Applying \cref{gen::lma::decomp_to_sched_computation} gives us an $O(1/\eps)$-batch $(1+\eps,x+O(k/\eps))$-approximation schedule. Recall that $x=\eps^2|\alpha\oplus\beta| \leq 2\eps^2\max\{|\alpha|,|\beta|\} = 2\eps^2M$, and $k/\eps=\max\{1,2 \eps^2 M\}/\eps\le 1/\eps+2\eps M$, and so, the approximation can be rewritten as $(1+O(\eps),O(1/\eps))$. Rescaling $\eps$ gives the desired result.
\end{proof}

\subsection{Computing Small Separator Decompositions for General Graphs}
To prove \cref{thm:genGraphSepDecomp}, we first compute a small separator decomposition with no bound on the number of clusters (\Cref{thm:ballCarving}) and then in a second step use the results from \Cref{ssec:clusterMerging} to reduce the number of clusters. Both steps do not make any assumption on a respective graph class. 
\begin{theorem} \label{thm:ballCarving}
	For any $\epsilon>0$ there is a deterministic \LOCAL algorithm that, given a (possibly disconnected) graph $G=(V,E)$, computes a strong $(d,x)$-separator decomposition with $d=O(\log n/\eps)$ and $x=\eps n$  in $O(\poly\log (n)/\eps)$ rounds.
\end{theorem}

\subparagraph*{SLOCAL model.} The proof of \cref{thm:ballCarving} goes through an \SLOCAL (explained hereafter) algorithm and its conversion into a LOCAL algorithm with the desired run-time.
In the \SLOCAL model, nodes of a graph are processed in an arbitrary (adversarial) sequential order $v_1,v_2,\dots,v_n$. In an \SLOCAL algorithm with locality $r$, the output of a vertex is a function from its $r$-neighborhood in the graph, including outputs of already processed vertices. E.g., a $(\Delta+1)$-vertex coloring can be computed by an SLOCAL algorithm with locality $1$ by iterating through the vertices in an arbitrary order and greedily choosing a color that is not used by an already colored neighbor. The computed solution depends on the order in which vertices are processed. The Achilles' heel of a correct \SLOCAL algorithm for some problem $\Pi$ is that it has to compute a feasible solution regardless of the order, even if it is picked adversarially. In the classic \SLOCAL model vertices cannot determine the output of other nodes, but it was already observed in \cite{SLOCAL17} that one can let vertices determine the output of vertices in their $R$-hop neighborhood at the cost of increasing the locality by at most a constant factor. Thus,  ball carving algorithms like the one by Awerbuch et al. \cite{awerbuch89} and also the algorithm to be presented in the proof of \cref{thm:ballCarving} are in fact \SLOCAL algorithms. Ghaffari et al. \cite{SLOCAL17} show that an \SLOCAL algorithm with locality $R$ can be simulated in the \LOCAL model in time $O(R\cdot (d+1)\cdot c)$ given a $(d,c)$-network decomposition of $G^{O(R)}$. 
\begin{definition}[Network Decomposition, \cite{awerbuch89}]
\label{def:decomposition}
  A weak \emph{$\big(d(n),c(n)\big)$-network decomposition} of an
  $n$-node graph $G=(V,E)$ is a partition of $V$ into clusters such
  that each cluster has weak  diameter at most $d(n)$ and the
  cluster-graph is properly colored with colors $1,\dots,c(n)$.
\end{definition}
As such a decomposition with $d=c=O(\log n)$ can be computed in $O(R\cdot \log^2 n)$ rounds with the randomized algorithm by Linial and Saks \cite{linial93} and deterministically in time $O(R\cdot \log^7 n)$ with the algorithm from \cite{Vaclav_2019}, the following result holds. 
\begin{theorem}[\SLOCAL in \LOCAL, \cite{SLOCAL17,Vaclav_2019,GGR20}] 
\label{thm:SLOCAL}
We can run an \SLOCAL algorithm with locality $R$ in $O(R\cdot \log^2 n)$ randomized time and in  $O(R\cdot \log^7 n)$ deterministic time in \LOCAL.
\end{theorem}

We use \cref{thm:SLOCAL} to prove \cref{thm:ballCarving}.
\begin{proof}[Proof of \cref{thm:ballCarving}]
We  describe a sequential algorithm to compute such a decomposition that is inspired by a sequential construction inspired by \cite{awerbuch89,linial93,SLOCAL17}. This algorithm can be turned into a distributed algorithm with \cref{thm:SLOCAL}. 
\subparagraph*{Ball carving in $G$.}
We describe a sequential ball carving process in $G$. We sequentially iterate through the (remaining) vertices in an arbitrary order; assume that we have constructed clusters $C_1,\dots, C_{i-1}$ and that we are currently processing a vertex $v$ with the objective to construct cluster $C_i$. We compute the smallest radius $r_v$ such that 
\begin{align}
\label{eqn:ballGrowing}
|B_{r_v+1}(v)|\leq (1+\epsilon)|B_{r_v}(v)|
\end{align}
holds, where the ball around $v$ is computed in the graph induced by the remaining vertices of $G$. 
Then we set $C_i=B_{r_v}(v)$, add $B_{r_v+1}(v)\setminus B_{r_v}(v)$ to $S$ and remove $B_{r_v+1}(v)$ from the graph. 
\begin{claim}
 There always exists an integer $r_v=O(\epsilon^{-1}\log n)$ satisfying \cref{eqn:ballGrowing}.
\end{claim}
\begin{proof}
We have that $|B_0(v)|=1$, and whenever all radii up to $r$ do not satisfy  \cref{eqn:ballGrowing} we have 
\begin{align*}
|B_{r+1}(v)|\geq (1+\eps)|B_r(v)|\geq\dots\geq (1+\eps)^{r}|B_0(v)|= (1+\eps)^{r}.
\end{align*}
As $|B_{r+1}(v)|\leq n$ the condition cannot be violated for all $r\in O(\log_{1+\eps}n)$ and we have that $r_v=O(\log_{1+\eps}n)=O(\eps^{-1}\log n)$.
\renewcommand{\qed}{\ensuremath{\hfill\blacksquare}}
\end{proof}
\renewcommand{\qed}{\hfill \ensuremath{\Box}}

\begin{claim}
At the end of the process $C_1,\dots,C_k, S$ gives a strong $(d,x)$-separator decomposition with empty ghost node set and $x=\eps|n|$.
\end{claim}
\begin{proof}
At the end of the process all vertices in $V$ are clustered in some $C_i$ or contained in $S$. Let $S_i$ be the nodes that are added to $S$ when constructing $C_i$ by forming a ball of radius $r_v$ around $v$. Due to  \cref{eqn:ballGrowing} we can bound $|S_i|=|B_{r_v+1}(v)\setminus B_{r_v}(v)|\le \eps|C_i|$, that is, we can bound $|S|=|S_1\cup \dots \cup S_k|$ by $\eps(|C_1|+\dots+|C_k|)\leq \eps |V|$. The separation property of $S$ follows as each cluster $C_i$ is separated from all clusters $C_j$ with $j>i$ by the vertices in $S_i\subseteq S$. \renewcommand{\qed}{\ensuremath{\hfill\blacksquare}}
\end{proof}
\renewcommand{\qed}{\hfill \ensuremath{\Box}}
This ``sequential'' algorithm can be implemented in $\poly\log (n) /\eps$  rounds in the distributed setting via \cref{thm:SLOCAL}.
\end{proof}

\subparagraph*{Small separator decomposition on general graphs.} We are now ready to prove \cref{thm:genGraphSepDecomp} using \cref{thm:ballCarving}.

\begin{proof}[Proof of \cref{thm:genGraphSepDecomp}]
Let $G=(V,E)$ and let $\alpha$ and $\beta$ be two vertex covers of $G$. We first create a small separator decomposition of $G[\alpha\oplus\beta]$, and later extend it to a decomposition of the whole graph.

Apply \cref{thm:ballCarving} to compute a strong $(d,x)$-separator decomposition of $G[D]$ with empty ghost node set and separator set $S$ in $O(\poly\log n/\eps)$ rounds. Now, apply \cref{lem:newSeparatorDecomp} to obtain a $(d,x)$-separator decomposition of $G$ with ghost node set $Z=V\setminus (\alpha\oplus \beta)$, separator set $S$ and parameters $d=O(\log n/\eps)$, $x=\eps |\alpha\oplus\beta|$ and $k= \max\{1,2 \eps M\}$.
	
The run-time is dominated by the first step which takes $O(\poly\log n/\eps)$ rounds.
\end{proof}

\section{Distributed Computation of Schedules III (Bounded Arboricity)}
\label{ssec:overshooting}
In this section we prove the following theorem for the distributed computation of reconfiguration schedules on graphs with bounded arboricity. 

\arbOvershooting*

Note that even though the assumption on the knowledge of the arboricity is somewhat unusual, the nodes might be able to get it indirectly in some cases, for example, if they know that the graph is planar, they also know that the arboricity is bounded by $3$. In fact, since we are only interested in the arboricity of the bipartite graph $G[\alpha\oplus \beta]$, they know that $\lambda\le 2$ as triangle free planar graphs have at most $2n-3$ edges. We obtain the following corollary.

\begin{corollary}[Planar Graphs]
	\label{cor:planar}
	Let  $G=(V,E)$ be a planar graph with two vertex covers $\alpha$ and $\beta$. There is a  $O\parenths*{\log^* (n)/\varepsilon}$-round \LOCAL algorithm that computes an $O(1/\eps^2)$-batch reconfiguration schedule from $\alpha$ to $\beta$ that is a $(2-\frac{1}{4}+\eps,1)$-approximation.
\end{corollary}

In order to prove \Cref{thm:arbOvershooting}, we first present a simple algorithm that yields an  $O(\Delta/\eps)$-batch schedule (\cref{ssec:longDegree}).  In \cref{ssec:shortDegree}, we show how the bounded arboricity case can be reduced to a bounded degree case and obtain a schedule of length $O(\lambda/\eps^2)$.

In both parts, $G = (V,E)$ is the input graph and $\alpha$ and $\beta$ are the two input vertex covers. Let $A = \alpha\setminus\beta$, by $B = \beta\setminus\alpha$ and by $M=\max\{|\alpha|,|\beta|\}$.
\subsection{An $O(\Delta/\eps)$-Batch Schedule}
\label{ssec:longDegree}
In this section, we assume  the graph $G$ is arbitrary (i.e., no knowledge of $\lambda$).

We call a partition of the vertices in $A$ into clusters $\mathcal{E}_1,\hdots,\mathcal{E}_\ell$ \emph{degree ordered} if  
\begin{itemize}
\item for every two nodes $v_i \in \mathcal{E}_i$, $v_j \in \mathcal{E}_j$ with $i < j$, $deg(v_i) \leq deg(v_j)$, where $deg$ refers to the degree in the induced subgraph $G[\alpha\oplus\beta]$, and
\item $|\mathcal{E}_i| = O(\eps M)$ where $M = \max\{|\alpha|,|\beta|\}$  and $\eps$ is some small (tunable) parameter. 
\end{itemize}
Note that the partition is defined for given parameters $\ell,\eps$, and vertex covers $\alpha$ and $\beta$. 
We associate the following schedule with a given degree ordered partition. 

\subparagraph*{$O(\ell)$-batch schedule (for a given degree ordered partition).}
At step $2i-1$, add $N(\mathcal{E}_i)\cap B$ to the vertex cover. At step $2i$, remove $\mathcal{E}_i$ from the vertex cover.

In \cref{lem:degBatchCost} we bound the length and the cost of the schedule. Later, in \cref{lem:degBasedClustering} we show how to compute the needed partition with  an adequate choice of parameters ($\ell$ and $\eps$). 
Iterating through the $\mathcal{E}_i$'s and reconfiguring them sequentially, with an arbitrary order inside each $\mathcal{E}_i$, yields the same schedule as described in \cref{ssec:degbased}, which is the main ingredient for the next lemma.
\begin{lemma}[Schedule Cost]\label{lem:degBatchCost}
Let $\eps \in (0,1), \xi\geq 1, c\geq 0$ and let $G=(V,E)$ be a graph with two vertex covers $\alpha$ and $\beta$ for which the sequential greedy schedule yields a $(\xi,c)$-approximation, regardless of the order of vertices inside a degree class. 
Then any schedule associated with a degree ordered partition into $\ell$ clusters and parameter $\eps$  is an $O(\ell)$-batch $(\xi + \eps,c)$-approximation schedule.
\end{lemma}
\begin{proof}
Let $\mathcal{S}$ be the batch schedule associated with the given partition $\mathcal{E}_1,\dots,\mathcal{E}_{\ell}$. Schedule $\mathcal{S}$ has length $O(\ell)$  and is valid as $A=\alpha\setminus\beta$ and $B=\beta\setminus\alpha$ are independent sets, and a node is only removed from the cover after all of its neighbors have been added (so all edges are always covered). 

To bound the cost of $\mathcal{S}$, let $\mathcal{S}'$ be an arbitrary sequential execution of $\mathcal{S}$, i.e, the  reconfiguration step `Remove $\mathcal{E}_i$' is replaced with $|\mathcal{E}_i|$ single steps in which the vertices in $\mathcal{E}_i$ are removed in an arbitrary order, and `Add $N(\mathcal{E}_i)\cap B$' is replaced with $|N(\mathcal{E}_i)\cap B|$ single steps, accordingly. By the assumption in the lemma $\mathcal{S}'$ is a $(\xi,c)$-approximation.

Let $1\leq s\leq \ell$ and let $W_s$ be the vertex cover right before removing $\mathcal{E}_s$ in $\mathcal{S}$. Observe that 
\[
W_s=\bigcup_{r=s}^{\ell} \mathcal{E}_r\cup \left(N(\bigcup_{r=1}^{s} \mathcal{E}_r)\cap B\right)\ .
\]
Next, consider the vertex cover $W'_s$ obtained in the greedy schedule $\mathcal{S}'$, just after removing the last vertex of $\mathcal{E}_{s}$. We have 
\[
W'_s=\bigcup_{r=s+1}^{\ell}\mathcal{E}_r\cup \left(N(\bigcup_{r=1}^{s} \mathcal{E}_r)\cap B\right)\ .
\]
Schedule $\mathcal{S}'$  achieves cost at most $\xi M + c$. Thus, $|W'_s| \leq \xi M + c$. We bound the size of $W_s$ by
\begin{align*}
    |W_s| = |W'_s|+|W_s \setminus W'_s| = |W'_s| + |\mathcal{E}_s| \stackrel{(*)}{\leq} (\xi + O(\eps)) M + c.
\end{align*}
where we used $|\mathcal{E}_i|=O(\eps M)$ (by the definition of the partition) at $(*)$. Rescaling $\eps$ implies the claimed bound on the cost of our schedule.
\end{proof}

Now, we show how to compute a suitable degree ordered partition of $A$.

\begin{lemma}[Degree-Based Clustering]\label{lem:degBasedClustering}
Let $\eps \in (0,1)$ and let $G = (V,E)$ be a graph with two vertex covers $\alpha$ and $\beta$. There exists a \LOCAL algorithm that finds a degree ordered partition $\mathcal{E}_1,\dots,\mathcal{E}_{\ell}$ of $A$ with parameter $\eps$ and length $\ell=O(\Delta/\eps) $ in  $O(\logstar(n)/\eps)$ rounds.
\end{lemma}

\begin{proof}
For $1\le i\le\Delta$, let $U_i$ be the set of degree-$i$ vertices of $A$ (in $G[A\cup B]$). Note that $A=\bigcup_{i=1}^\Delta U_i$.
Starting with the trivial clustering of $G$ in which each vertex forms a singleton cluster, apply the algorithm from \cref{lem:compressingClustering} with parameter $\eps$ to compute a clustering $\mathcal{C}$. The clustering consists of at most $\ceil{2\eps M}$ clusters, since if there is more than one cluster, every cluster has diameter at least $1/\eps$, and contains at least $1/(2\eps)$ nodes from any vertex cover of $G$ (since any vertex cover of a path contains at least half of its vertices).

Within each cluster $C\in \mathcal{C}$, we compute the following partition of $C \cap A$. 
For each $1\le i\le \Delta$, partition $U_i\cap C$ into exactly $t=\lceil 1/\eps\rceil$ subsets $A^C_{i,1},A^C_{i,2},\dots,A^C_{i,t}$, each of size at most $\ceil*{\eps|U_i \cap C|}$ (some of them might be empty). 
Let $A_{i,j} = \bigcup_{C\in\mathcal{C}}A^C_{i,j}$, we obtain
\begin{align*}
    |A_{i,j}| = \sum_{C\in\mathcal{C}} |A^C_{i,j}| \leq  \sum_{C\in\mathcal{C}} (\eps|U_i \cap C| + 1) = \eps|U_i| + |\mathcal{C}| = O(\eps M)\ ,
\end{align*}
where the last equality follows from the facts that $U_i \subseteq A$ and $|\mathcal{C}| = O(\eps M)$. 

We obtain the desired partition by ordering the sets $A_{i,j}$ lexicographically $A_{1,1},A_{1,2},\hdots,A_{1,t},\allowbreak A_{2,1}, \allowbreak A_{2,2},\hdots$, formally, we set $\mathcal{E}_k = A_{i,j}$ for $k = (i-1)\cdot t + j$ for $k=1,\dots,\Delta\cdot t$. The length of the partition is $\ell=\Delta\cdot t=O(\Delta/\eps)$. 

The runtime bound follows,  as applying \cref{lem:compressingClustering} takes $O(\logstar(n)/\eps$) rounds, and for each node to know its corresponding $i,j$ (and to decide on the index in the partition), it just needs to learn its cluster $C$, which is done in additional $O(1/\eps)$ rounds.
\end{proof}

We summarize the results of this section in the following theorem.
\begin{theorem}\label{thm:deltabased}
Let $G = (V,E)$ be an $n$-node graph with two vertex covers $\alpha$ and $\beta$, for which the sequential greedy schedule is an $(\eta,c)$-approximation.
For every $\eps>0$, there exists an $O\parenths*{\log^* (n)/\eps}$-round \LOCAL algorithm that computes an $O\left(\Delta/\eps\right)$-batch $(\eta+\eps,c)$-approximation  schedule from $\alpha$ to $\beta$.
\end{theorem}

\begin{proof}
Applying the described schedule with an appropriate scaling of $\eps$, along with \cref{lem:degBatchCost,lem:degBasedClustering} yields the desired result.
\end{proof}

Applying \cref{thm:deltabased} to graphs of bounded arboricity $\lambda$, together with the sequential $(2-1/(2\lambda),1)$-approximation schedule from \cref{thm:seqmindeg} yields the following result. 
\begin{corollary}\label{cor:degBasedBoundedArb}
Let $G = (V,E)$ be a graph with arboricity $\lambda$ and two vertex covers $\alpha$ and $\beta$. For every $\eps > 0$, assuming that nodes know the maximum degree $\Delta$ in $G[\alpha\oplus\beta]$, there exists an $O(\logstar(n)/\eps)$-round \LOCAL algorithm that computes an $O(\Delta/\eps)$-batch $\left(2-\frac{1}{2\lambda}+\eps,1\right)$-approximation schedule from $\alpha$ to $\beta$.
\end{corollary}

\subsection{From Bounded Degree to Bounded Arboricity: an $O(\lambda/\eps^2)$-Batch Schedule}
\label{ssec:shortDegree}
\cref{cor:degBasedBoundedArb} gives  a $(2-1/(2\lambda)+\eps,1)$-approximation schedule of length $O(\Delta/\eps)$ on graphs with arboricity $\lambda$. However,  graphs of very low arboricity (such as planar graphs) might have a very large maximum degree, making the schedule inefficient in terms of schedule length.
We will now describe a batch schedule for graphs of bounded arboricity which uses the results of the previous section as a subroutine. It requires global knowledge of  $\lambda$ and $\eps$. We use the notation $\eta= \lceil2\lambda/\eps\rceil$ to describe the schedule.

\subparagraph*{$O(\lambda/\eps^2)$-batch schedule.} In the first batch, add all vertices in $B_\eta=\{v\in B : |N(v)\cap A|\le \eta\}$ to the vertex cover. Then apply the algorithm from \cref{thm:deltabased}  with $\Delta=\eta$ and $\eps'=\eps/2$ to compute a schedule from $B'=B\setminus B_{\eta}$ to $A'=A\setminus \{v\in A : N(v)\cap B\subseteq B_\eta\}$. Reverse the obtained schedule via \cref{obs:reverseschedule}, to complete the $A$-to-$B$ schedule.

\smallskip

We will use the following result to bound the cost of the schedule in the proceeding theorem. 
\begin{lemma}[{\cite[Lemma~8]{Solomon18}}]\label{lem:shay}
Let $W$ be a vertex cover in a graph $G=(V,E)$ of arboricity $\lambda$, and $H=\{v\in V: d_G(v)\ge \lambda/\tau\}$, for a given $\tau\in (0,1)$. Then $|H\setminus W|\le \eps |W|$. 
\end{lemma}

\begin{proof}[Proof of \Cref{thm:arbOvershooting}]
The validity of the schedule follows from the validity of the $B'$-to-$A'$ schedule and \cref{obs:reverseschedule}.
Thus, from \cref{thm:deltabased}, the runtime of the algorithm is $O(\log^*(n)/\eps)$ (note that reverting the schedule can be done by each node locally as the $O(\eta/\eps)$ global upper bound on the length of the schedule is globally known). We apply \cref{lem:shay} with the vertex cover $A'$ in the graph $G[A'\cup B']$ and $\tau=\lambda/\eta$, to see that $|B_\eta|\le (\lambda/\eta)|A|$. Thus, the first batch only contributes 1 to the length of the schedule and a $\lambda|A|/\eta$ term to the cost. By \cref{thm:deltabased} and \cref{thm:seqmindeg}, the  $A'$-to-$B'$ schedule we computed is an $O(\eta/\eps)=O(\lambda/\eps^2)$-batch schedule and has cost at most $(2-1/(2\lambda)+\eps/2)M+1$, where $M=\max(|\alpha|, |\beta|)$. Thus, the overall cost is $(2-1/(2\lambda)+\eps/2 + \lambda/\eta)M+1$, which gives the promised approximation, by the choice of $\eta$.
\end{proof}

\section{Lower Bounds for Distributed Computation of Schedules}
\label{app:distrLowerbounds}
First, in \Cref{subsec::weighted_lb}, we present a \emph{gap} between the existence of reconfiguration schedules of weighted vertex covers and their computation. That is, we show that there exists a family of weighted cycles on which there always exists a $(1+\eps)$-approximation schedule, but computing it takes $\Omega(n)$ rounds.
Then, in \Cref{subsec::unweightedLB} we lift the lower bound to the unweighted case. Finally, in  \cref{subsec::unweightedLBcycles}, we turn our focus to unweighted cycles, and show that computing a $(2-\eps)$-approximation schedule  with $O(1)$ batches requires $\Omega(\logstar n)$ rounds.

\subsection{An $\Omega(n)$ Lower Bound for Weighted Cycles}
\label{subsec::weighted_lb}

In the \emph{(minimum) weighted vertex cover} problem each vertex $v$ is assigned a \emph{weight} $w(v)\in \NN$ and the objective is to find a vertex cover $S$ of minimum total weight $\sum_{v\in S} w(v)$. 
The weighted \emph{vertex cover reconfiguration problem} is a natural generalization of the unweighted vertex cover reconfiguration problem, that is, the only change is that the cost of a schedule is determined by the total weight of the (intermediate) vertex covers. It is known that there are graph families on which even $(2-\eps)$-approximation schedules do not exist for constant $\eps>0$, e.g., complete bipartite graphs.
In this section, we show that there are graphs on which $(1+\eps)$-approximation schedule exists and can also be computed efficiently in the centralized setting, but computing such schedules in the \LOCAL model takes $\Omega(n)$ rounds.

Let $c \geq 1$ be a constant integer and $t: \NN \to \NN$ be a function such that $t(n) = o(n)$. Next, we define a graph family  $\mathcal{G}_{t,c}= \{G_{n,t,c}\mid n\in\NN\}$. 
We omit $t$ and $c$ when they are clear from context.

\subparagraph*{Graph $G_{n,t,c}$.} 
Let $G_{n,t,c}=(V,E,w)$ be an $n$-node cycle consisting of $k := \frac{n}{4t(n) + 2}$ 
identical paths, $I_1, \hdots, I_k$, where
\begin{itemize}
    \item $I_i$ is a path $v_1^i,v_2^i, \hdots, v_{4t(n)+1}^i, v_{4t(n)+2}^i$ of $4t(n)+2$ nodes;
    \item \emph{heavy} nodes: $w(v_{2t(n)+1}^i)=w(v_{2t(n)+2}^i)=c\cdot t(n)$; 
    \item \emph{light} nodes: $w(v_j^i)=1$, for $1\le j\le 4t(n)+2$.
    \item The paths are connected with edges  $(v_{4t(n)+2}^i,v_{1}^{i+1})$, $1\le i \le k-1$, and $(v_{4t(n)+2}^k, v_{1}^{1})$.
\end{itemize}

\subparagraph*{Properties of $G_{n,t,c}$.} Note that this graph has two optimal vertex covers: 
\begin{align*} 
	V_{odd} & := \bigcup_{i=1}^k\{v_j^i \ \vert \ j \text{ is odd}\}, \text{ and } 
	 V_{even}  := \bigcup_{i=1}^k\{v_j^i \ \vert \ j \text{ is even}\}.
\end{align*}
Note that when restricted to any single segment $I_i$, both $V_{odd}$ and $V_{even}$ have exactly $1$ heavy node and $2t(n)$ light nodes, hence, using the notation $w(S)=\sum_{v\in S} w(v)$, we have $w(V_{odd}\cap I_i) = w(V_{even}\cap I_i) = (c+2)t(n)$, and $w(V_{odd}) = w(V_{even}) = k(c+2)t(n) = \Theta(n)$.

\begin{lemma}
\label{goodBatchSched}
Let $c\in\NN$ and $t(n) = o(n)$. For every $G = G_{n,t,c} \in \mathcal{G}$ and  pair $\alpha, \beta$ of vertex covers of $G$, there is a monotone $(1,O(t(n)))$-approximation schedule of length $O\parenths*{\frac{n}{t(n)}}$.
\end{lemma}

\begin{proof}
Let $C_1,\hdots, C_r$ be the set of connected components in $G[\alpha\oplus\beta]$. First, we reduce the case $r=1$ to the case $r>1$, and then prove the lemma for the latter. If $r=1$, add two weight-$1$ $\beta$-nodes, $b_1$ and $b_2$, to $\alpha$ and consider the following reasoning given a schedule from $\alpha \cup \{b_1,b_2\}$ to $\beta$: if $\mathcal{S}$ is an $(\eta,c)$-approximation schedule, one can obtain an $(\eta,c+2)$-approximation schedule $\mathcal{S}'$ from $\alpha$ to $\beta$ (add $\{b_1,b_2\}$ to the cover, then apply $\mathcal{S}$). Therefore, as $G[(\alpha\cup\{b_1,b_2\})\oplus\beta]$ has at least $2$ connected components, we can only reason about the case $r>1$.

Now, consider $r>1$ and denote by $\mathcal{C}_{\alpha,h}$ the family of all the sets $C_i$ for which $w(C_i \cap \alpha) > w(C_i \cap \beta)$ and $C_i$ contains a heavy node and denote by $\mathcal{C}_{\alpha,l}$ the family of all the sets $C_i$ for which $w(C_i \cap \alpha) > w(C_i \cap \beta)$ and $C_i$ contains only light nodes, similarly, define $\mathcal{C}_{\beta,h}$ and $\mathcal{C}_{\beta,l}$ for sets $C_i$ with $w(C_i \cap \alpha) \leq w(C_i \cap \beta)$. Abusing the above notation, we will occasionally refer to, e.g.,  $\mathcal{C}_{\alpha,l}$, as a set of nodes.

We now prove that every set $\mathcal{C}_{\gamma,\xi}$ ($\gamma \in \{\alpha,\beta\}, \xi\in\{l,h\}$) admits a $(1,O(t(n)))$-approximation schedule of length $O\big(\frac{n}{t(n)}\big)$
between $\alpha$ and $\beta$ on the graph induced by $G[\mathcal{C}_{\gamma,\xi}]$. With this, see that a schedule which reconfigures $\mathcal{C}_{\alpha,l}$, $\mathcal{C}_{\alpha,h}$, $\mathcal{C}_{\beta,l}$ and $\mathcal{C}_{\beta,h}$,  in this order, is a valid schedule as there are no edges between $G[\mathcal{C}_{\gamma,\xi}]$ and $G[\mathcal{C}_{\gamma',\xi'}]$ if $\gamma\neq \gamma'$ or $\xi\neq \xi'$. To see that the global approximation guarantee follows, note the following
\begin{enumerate}
    \item During the reconfiguration of $\mathcal{C}_{\alpha,l}$, the weight of the cover is at most $w(\alpha) + O(t(n))$.
    \item During the reconfiguration of $\mathcal{C}_{\alpha,h}$, the weight of the cover is at most
    \begin{align*}
        w(\beta \cap \mathcal{C}_{\alpha,l}) + w(\alpha \setminus \mathcal{C}_{\alpha,l}) + O(t(n)) \leq w(\alpha) + O(t(n)).
    \end{align*}
    The inequality follows from the fact that $w(\beta \cap \mathcal{C}_{\alpha,l}) \leq w(\alpha \cap \mathcal{C}_{\alpha,l})$.
    \item During the reconfiguration of $\mathcal{C}_{\beta,l}$, the weight of the vertex cover is at most 
    \begin{align*}
        w(\beta \cap \mathcal{C}_{\alpha}) +
        w(\alpha \cap \mathcal{C}_{\beta,h}) +
        w(\beta \cap \mathcal{C}_{\beta,l}) +
        O(t(n)) \leq w(\beta) + O(t(n)),
    \end{align*}
    where $\mathcal{C}_{\alpha} = \mathcal{C}_{\alpha,l} \cup \mathcal{C}_{\alpha,h}$. Note that the inequality follows from the fact that $w(\alpha \cap \mathcal{C}_{\beta,h}) \leq w(\beta \cap \mathcal{C}_{\beta,h})$.
    \item Lastly, during the reconfiguration of $\mathcal{C}_{\beta,h}$, the weight of the cover is at most $w(\beta)  + O(t(n))$.
\end{enumerate}

\subparagraph*{Schedule for $\mathcal{C}_{\alpha,l}$ and $\mathcal{C}_{\beta,l}$.} To prove that all families have a good reconfiguration schedule, first note that both $\mathcal{C}_{\alpha,l}$ and $\mathcal{C}_{\beta,l}$ are forests of nodes with weight $1$. Therefore, the reconfiguration task is the same as the one for unweighted vertex covers, and so, applying \cref{thm:variousClassesBatch} with $\eps=t(n)/n=o(1)$, we have a $(1,O(t(n)))$-approximation schedule of length $O\parenths*{\frac{n}{t(n)}}$ on $\mathcal{C}_{\alpha,l}$ and $\mathcal{C}_{\beta,l}$. 

\subparagraph*{Schedule for $C_i \in \mathcal{C}_{\alpha,h}$.} For $\mathcal{C}_{\alpha,h}$, consider the segments $I_1,\hdots,I_k$ of $G$ as described above.
Let $C_i \in \mathcal{C}_{\alpha,h}$ and denote by $k_i = |\{ j\in[k] \mid I_j\cap C_i \neq \emptyset\}|$ the number of segments $C_i$ intersects. Note that every $C_i$ is a path alternating between $\alpha$-nodes and $\beta$-nodes, and so, we can assume without loss of generality that $\alpha \cap C_i = V_{odd}\cap C_i$ and $\beta \cap C_i = V_{beta}\cap C_i$. Let $I_{j_1},\hdots,I_{j_{k_i}}$ be the segments intersecting $C_i$ and consider the following schedule:
\begin{itemize}
    \item $V_0 = \alpha \cap C_i$.
    \item For $s = 1,\hdots,k_i$, $V_{2s-1} = V_{2s-2} \cup (\beta_{j_s} \cap C_i)$, $V_{2s} = V_{2s-1} \setminus (\alpha_{j_s} \cap C_i)$.
\end{itemize}

\noindent\textit{Properties of the schedule of $C_i$:} The validity of the schedule follows from the fact that $v_{2t(n)+2}^{j_s} \in \beta$ for all $s = 1,\hdots,k_i-1$. For the approximation, see that in $w(V_0 \oplus V_1) \leq w(\beta \cap I_{j_1}) = (c+2)t(n)$ and $(V_1 \oplus V_2)\subseteq \alpha$ (i.e., only removing nodes), implying that $w(V_2) \leq w(\alpha\cap C_i) + (c+2)t(n)$. For $s = 2,\hdots,k_i-1$ we have $w(V_{2s-1}\oplus V_{2s-2}) = w(V_{2s}\oplus V_{2s-1}) = (c+2)t(n)$, implying that $w(V_{2s-1}) \leq w(\alpha\cap C_i) + 2(c+2)t(n)$. Finally, $w(V_{2k_i-1} \oplus V_{2k_i-2}) \leq  w(\beta \cap I_{j_{k_i}}) = (c+2)t(n)$ and $(V_{2k_i} \oplus V_{2k_i-1}) \subseteq \alpha$, giving a schedule with a cost bounded by $3(c+2)t(n) = O(t(n))$ and a length of $O(k_i)$.

\noindent\textit{Combining $C_i\in \mathcal{C}_{\alpha,h}$ Schedules:} Consider the schedule of $\mathcal{C}_{\alpha,h}$ in which we reconfigure the $C_i$, one after the other, where we reconfigure $C_i$ in $O(k_i)$ batches, totalling in $\sum_{i} O(k_i)$ batches. As $w(C_i \cap \beta) \leq w(C_i \cap \alpha)$ for all $C_i \in \mathcal{C}_{\alpha,h}$, before the reconfiguration of $C_i$ the weight of the cover is at most $w(\mathcal{C}_{\alpha,h}\cap \alpha)$. In addition, since every $C_i$ admits a $(1,O(t(n)))$-approximation schedule (and $w(C_i \cap \beta) \leq w(C_i \cap \alpha)$), during the reconfiguration of $C_i$ the weight of the cover is at most $w(\mathcal{C}_{\alpha,h}\cap \alpha) + O(t(n))$. This means that the computed schedule is a $(1,O(t(n)))$-approximation schedule of $\mathcal{C}_{\alpha,h}$. The length of the schedule is $\sum_{i}O(k_i)$. To bound this expression, see that every segment $I_j$ can intersect at most $4$ connected components in $\mathcal{C}_{\alpha,h}$. Indeed, if $I_j$ intersects $5$ connected components $C_{1},\hdots,C_{5}$, then at least $3$ of them are contained in $I_j$. As $I_j$ only contains only $2$ heavy nodes, at least one $C_i$ does not contain a heavy node. This is a contradiction, as $C_i \in \mathcal{C}_{\alpha,h}$ if and only if $C_i$ contains a heavy node. Therefore, we conclude that $\sum_{i}O(k_i) = O(k) = O\parenths*{\frac{n}{t(n)}}$.  

\subparagraph*{Schedule for $\mathcal{C}_{\beta,h}$.}  To complete the proof, we reason about the existence of a schedule from $(\beta\cap \mathcal{C}_{\beta,h})$ to $(\alpha \cap \mathcal{C}_{\beta,h})$. Note that $w(C_i\cap\alpha) \leq w(C_i\cap\beta)$ for all $C_i \in \mathcal{C}_{\beta,h}$, therefore, a similar reasoning as 
for $\mathcal{C}_{\alpha,h}$ yields the same result. Then, by \cref{obs:reverseschedule}, reversing this schedule gives the desired schedule for $\mathcal{C}_{\beta,h}$.
\end{proof}

We now give our main results for this section. 
\begin{theorem}\label{thm::weighted_lb}
For every $\eps \in (0,1)$ and $t(n) = o(n)$, there exists a family $\mathcal{G}$ of weighted graphs, on which, for every large enough $n$, the following hold.
\begin{enumerate}
    \item Every $n$-node graph $G\in \mathcal{G}$ admits a monotone $(1+\eps)$-approximation schedule.
    \item On an ID space of size $\Omega(n^2)$, no $t(n)$-round algorithm can compute a monotone $(2 - \eps)$-approximation schedule of length $O(n)$ on $\mathcal{G}$.
\end{enumerate}
\end{theorem}

\begin{proof}
Let $\eps \in (0,1)$ and let $t(n) = o(n)$. Now, let $c\in\NN$ be some constant such that $c > \frac{2}{\eps}$ and take $\mathcal{G} = \mathcal{G}_{c,t}$ the graph family described above.

For the first part of the proof, note that \cref{goodBatchSched} gives a $(1,O(t(n)))$-approximation schedule for any constant $c>1$ and function $t(n) = o(n)$. As $M \geq k(c+2)t(n) = \Theta(n)$ ($M = \max\{w(\alpha),w(\beta)\}$), we have $t(n)/M = o(1)$. Thus, for every constant $\eps \in (0,1)$ and for any large enough $n$, a $(1,O(t(n)))$-approximation schedule implies that the size of all covers in the schedule is bounded by
\begin{align*}
    M + O(t(n)) = M + M\frac{O(t(n))}{M} < M(1+\eps)\ ,
\end{align*}
where the last inequality holds as $\frac{O(t(n))}{M} =o(1)< \eps$ for every large enough $n$.
Hence, it is indeed a $(1+\eps,0)$-approximation schedule.

For the second part of the proof, let $\mathcal{A}$ be an algorithm that computes a schedule of length $\ell(n)$ in $t(n)=o(n)$ rounds. Fix a large enough $n$ and let $t = t(n)$. Consider the $n$-node graph $G = G_{n,t,c}\in \mathcal{G}$ and let $\alpha = V_{odd}$ and $\beta = V_{even}$. Our goal is to show that there exists an ID assignment to nodes of $G$ such that, in the schedule computed by $\mathcal{A}$, all $k = \frac{n}{4t(n)+2}$ heavy $\beta$-nodes join the cover at the same time. Before their addition the vertex cover has cost at least $w(\alpha)=k(c+2)t(n)$, as $\alpha$ is a minimum vertex cover of the graph. Thus, after the addition of the $k$ $\beta$-nodes of cost $ct(n)$ each, we obtain a cost of $w(\alpha)+kct(n)$, which translates to a multiplicative approximation of at least
\begin{align*}
    \frac{w(\alpha)+kct(n)}{w(\alpha)} = 1 + \frac{kct(n)}{k(c+2)t(n)} = 1 + \frac{c}{c+2} > 2 - \frac{2}{c} > 2-\eps.
\end{align*}
Due to $c > \frac{2}{\eps}$ we obtain that $\mathcal{A}$ does not compute a $(2-\eps)$-approximation schedule on  $G$ and yields the desired result.

\subparagraph*{ID assignment, such that all $\beta$-nodes join at the same time.} To complete the proof, we have to show that there exists an assignment of IDs to the nodes of $G$, for which all heavy $\beta$-nodes join the cover at the same time. For this, note that in $t(n)$ rounds, each node learns only its $t(n)$-hop neighborhood, i.e., the $2t(n)+1$ nodes closest to it. Therefore, we can define a coloring $\mathcal{A} :
{[ID] \choose {4t(n)+2}} \to [\ell(n)]$ in the following way: 
Given a set $U = \{id_i\}_{i\in [4t(n)+2]}$ ordered in an ascending order (i.e., $id_1 < id_2 < \hdots < id_{4t(n)+2}$), $\mathcal{A}(U)$ is defined to be the time the heavy $\beta$-node $v^i_{2t(n)+2} \in I_i$ joins the cover when the IDs of the segment $I_i \subseteq G$ are in accordance to $U$, that is, the ID of $v_j^i$ is $id_j$ for all $j=1,\hdots,4t(n)+2$ (See \cref{fig:LBSegment}).

Recall, that the ID space is of size $\Omega(n^2)$ and see that $(4t(n)+2)k\ell(n) = O(n^2)$. Therefore, we can take $k\ell(n)$ pairwise disjoint sets of IDs $\{U_i\}_{i\in [k\ell(n)]} \subseteq {[ID] \choose {4t(n)+2}}$. Since $\mathcal{A}$ returns integers between $1$ and $\ell(n)$, by the pigeonhole principle, there are at least $k$ sets $U_{i_1},\hdots,U_{i_k}$ for which $\mathcal{A}(U_{i_1}) = \mathcal{A}(U_{i_2}) = \hdots = \mathcal{A}(U_{i_k})$. Therefore, if we set the IDs of $I_j$ in accordance to $U_{i_j}$ for all $j = 1,\hdots,k$, we get that in the computed schedule, all heavy $\beta$-nodes join the cover at the same time, giving the desired results.
\end{proof}

\begin{figure}[t]
    \centering
    \includegraphics{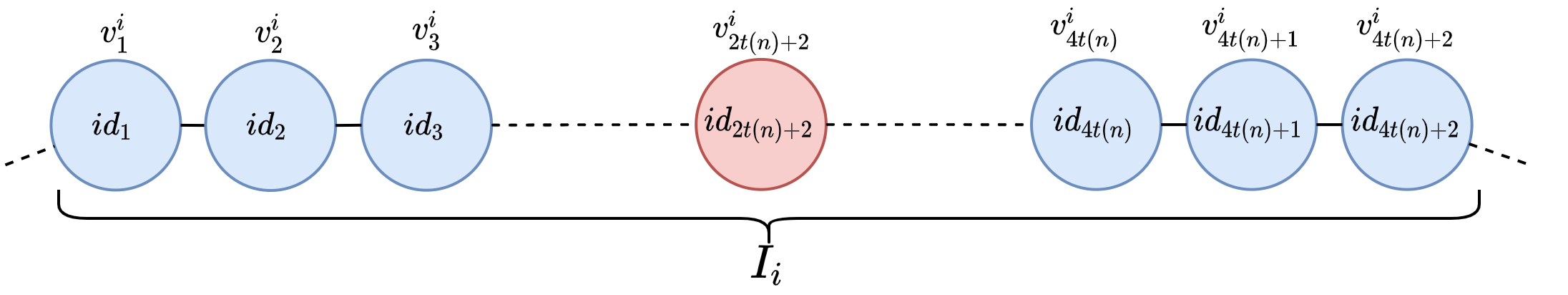}
    \caption{Segment $I_i$ with IDs in accordance to $U = \{id_i\}_{i=1}^{4t(n)+2}\subseteq {[ID] \choose 4t(n)+2}$. See that node $v_{j}^i$ has ID $id_j$ and that the red node, $v_{2t(n)+2}^i$, is the only heavy $\beta$-node in $I_i$. }
    \label{fig:LBSegment}
\end{figure}

\subsection{An $\Omega(n)$ Lower Bound for Unweighted Graphs}
\label{subsec::unweightedLB}
In this section, we reduce the task of schedule computation on weighted graphs to the same task on unweighted graphs. The lower bound for computing a reconfiguration schedule of weighted vertex cover from \cref{subsec::weighted_lb} combined with the proposed reduction, implies a lower bound for computing a reconfiguration schedule of unweighted vertex covers.

In the proposed reduction, given a weighted graph $G=(V,E,w)$ with $w : \NN \to \NN$, we replace each $v$ with $w(v)$ copies of $v$, and every edge $(u,v)$ is replaced by a complete bipartite graph $K_{w(u),w(v)}$ between all copies of $u$ and $v$. We continue with the formal definition of this graph.

\begin{definition}[Reduction Graph]
Let $G = (V,E,w)$ be a weighted graph with weights $w : V \to \NN$. We define the \emph{Reduction Graph} of $G$ as an unweighted graph $G' = (V',E')$ for which
\begin{align*}
    V' &= \bigcup_{v \in V}\curlybrackets*{v_i}_{i=1}^{w(v)}, \quad E' = \bigcup_{(u,v)\in E}\bigcup_{\substack{i \in [w(u)] \\ j \in [w(v)]}} \parenths*{u_i,v_j},
\end{align*}
where $[k] := \{1,2,\hdots,k\}$. That is, each vertex $v\in G$ has corresponding $w(v)$ vertices and every edge $(u,v)\in E$ has a corresponding  complete bipartite graph denoted by $G'_{u,v} \cong K_{w(u),w(v)}$.
\end{definition}

Given a weighted graph $G=(V,E,w)$ with weights $w:V\to \NN$ and its reduction graph $G' = (V',E')$ we denote by $corr(v) := \curlybrackets*{v_1,\hdots,v_{w(v)}} \subseteq V'$ the set of $w(v)$ nodes \emph{corresponding} to $v$. Given a subset $U \subseteq V$ we extend this notation and denote $corr(U) = \bigcup_{v\in U}corr(v)$. In addition, given a set $U' \subseteq V'$, we denote $corr^{-1}(U') = \{v\in V\mid corr(v) \subseteq U'\}$.
We first show a connection between vertex covers of a weighted graph $G$ and covers of its unweighted corresponding graph $G'$.

\begin{lemma}\label{reductionCoverEquivalence}
Let $G=(V,E,w)$ be a weighted graph with weights $w:V\to \NN$ and let $G'=(V',E')$ be the unweighted reduction graph of $G$. The following claims hold:
\begin{enumerate}
    \item For every vertex cover $U$ of $G$, the set $U' :=corr(U)$ is a vertex cover of $G'$ of size $|U'| = w(U)$.
    \item For every vertex cover $U'$ of $G'$, the set $U := corr^{-1}(U')$ is a vertex cover of $G$ of weight $w(U) \leq |U'|$.
\end{enumerate}
\end{lemma}

\begin{proof}
For the first part, see that $U' = \bigcup_{v \in U}corr(v)$ and that $w(v) = |corr(v)|$ for all $v\in V$ and so, $w(U) = \sum_{v\in U}w(v) = \sum_{v\in U}|corr(v)| = |U'|$. Now, let $(u',v')\in E'$ and let $u,v\in V$ be nodes such that $u' \in corr(u)$ and $v' \in corr(v)$. As $(u',v') \in E'$ we must have $(u,v)\in E$, and so, either $u \in U$ or $v \in U$. Therefore, we conclude that either $u' \in corr(u)$ or $v' \in corr(v)$ is in $U'$.

For the second part, see that $corr(v) \subseteq U'$ for every $v \in U$, and so, $w(U) = \sum_{v\in U}w(v) = \sum_{v\in U}|corr(v)| \leq |U'|$. To show that $U$ is a vertex cover, see that as every edge $(u,v)$ has a corresponding graph $G'_{u,v} \cong K_{w(u),w(v)}$ and as $U'$ is a vertex cover of $G'$, either $corr(u)\subseteq U'$ or $corr(v)\subseteq U'$ (otherwise, $G'_{u,v}$ would have not been covered).
\end{proof}

\begin{lemma}\label{lem::weighted_reduction}
Let $G=(V,E,w)$ be a weighted graph with weights $w : V \to \NN$, let $\alpha$ and $\beta$ be two vertex covers of $G$ and let $G' = (V',E')$ be the unweighted reduction graph of $G$. $G$ admits an $(\eta,c)$-approximation schedule of length $\ell$ between $\alpha$ and $\beta$ if and only if $G'$ admits a reconfiguration schedule between $corr(\alpha)$ and $corr(\beta)$ with the same length and approximation.
\end{lemma}

\begin{proof} We show both directions of the equivalence separately. 
\subparagraph*{From weighted to unweighted.}
Let $\mathcal{S} = (V_0 = \alpha, \hdots, V_\ell = \beta)$ be an $(\eta,c)$-approximation reconfiguration schedule between $\alpha$ and $\beta$. Define the following schedule $\mathcal{S}' = (V_0' = corr(V_0),\hdots, C_\ell'=corr(V_\ell))$. 
To show validity, first see that by \cref{reductionCoverEquivalence}, every $V'_i$ is a vertex cover. Next, to show that every $V_i' \oplus V_{i+1}'$ is an independent set, let $u',v' \in V_i' \oplus V_{i+1}'$ and let $u,v \in V$ be nodes such that $u' \in corr(u)$ and $v' \in corr(v)$. Note that $u,v\in V_i \oplus V_{i+1}$ (since $u',v' \in V_i' \oplus V_{i+1}'$) and so $(u,v) \notin E$, therefore, by the definition of the reduction graph, $(u',v') \notin E'$.

For the approximation factor, see that by \cref{reductionCoverEquivalence}, $|V_i'| = w(V_i)$ for all $i = 0,\hdots,\ell$, and in particular, the initial and final covers admit $w(\alpha) = |corr(\alpha)|$, $w(\beta) = |corr(\beta)|$. Therefore, as for every $i$, $w(V_i) \leq \eta\max\{w(\alpha),w(\beta)\} + c$, we must also have $|V_i'| \leq \eta\max\{|corr(\alpha)|,|corr(\beta)|\} + c$.
Therefore, the cost of $\mathcal{S'}$ is the same as of $\mathcal{S}$. Note that in addition, $|V'_0| = w(V_0)$ and $|V'_\ell| = w(V_\ell)$ and so, $\mathcal{S}'$ has an approximation factor of $(\eta,c)$ as well.

\subparagraph*{From unweighted to weighted.}
Now, let $\alpha$ and $\beta$ be vertex covers of $G$ and let $\mathcal{S}'=(V'_0=corr(\alpha) , \hdots , V'_\ell=corr(\beta))$ be a reconfiguration schedule between $corr(\alpha)$ and $corr(\beta)$. We define the schedule $\mathcal{S}$ between $\alpha$ and $\beta$ as $\mathcal{S}=\parenths*{V_0=corr^{-1}(V'_0),\hdots,V_{\ell} = corr^{-1}(V'_\ell)}$.
To show validity, first see that by \cref{reductionCoverEquivalence}, every $V'_i$ is a vertex cover. Next, to show that $V_i \oplus V_{i+1}$ is an independent set, assume for contradiction that there are $u,v \in V_i \oplus V_{i+1}$ such that $(u,v)\in E$. This implies that there is some $u' \in (V'_i \oplus V'_{i+1}) \cap corr(u)$ and some $v' \in (V'_i \oplus V'_{i+1}) \cap corr(v)$. This is a contradiction, as $(u',v')\in E'$ by definition of the reduction graph, but $V'_i \oplus V'_{i+1}$ must be an independent set since $\mathcal{S}'$ is a valid schedule.

For the approximation factor, see that by \cref{reductionCoverEquivalence}, $w(V_i) = |V_i'|$ for all $i = 0,\hdots,\ell$, and in particular, the initial and final vertex covers admit $|corr(\alpha)| = w(\alpha)$, $|corr(\beta)| = w(\beta)$. Therefore, as for every $i$, $|V_i'| \leq \eta\max\{|corr(\alpha)|,|corr(\beta)|\} + c$, we must also have $w(V_i) \leq \eta\max\{w(\alpha),w(\beta)\} + c$.
\end{proof}

In the following theorem, we lift the lower bound for weighted vertex cover reconfiguration from \Cref{subsec::weighted_lb} to unweighted vertex cover reconfiguration.

\begin{theorem}\label{thm::unwighted_gen_lb}
For every $\eps \in (0,1)$ and $t(n) = o(n)$, there exists a family $\mathcal{G}$ of unweighted graphs, on which, for every large enough $n$, the following hold.
\begin{enumerate}
    \item Every $n$-node graph $G\in \mathcal{G}$ admits a monotone $(1+\eps)$-approximation schedule.
    \item On an ID space of size $\Omega(n^2)$, no $t(n)$-round algorithm can compute a monotone $(2 - \eps)$-approximation schedule of length $O(n)$ on $\mathcal{G}$.
\end{enumerate}
\end{theorem}

\begin{proof}~
\subparagraph*{The graph family $\mathcal{G}$.} 
Let $\eps \in (0,1)$, let $t(n) = o(n)$ and let $h(n) = o(n)$ such that $t(n) = o(h(n))$ (e.g. $h(n) = t(n) \cdot \sqrt{n/t(n)}$). Now, let $\widetilde{\mathcal{G}}$ be the family of weighted graphs promised by \cref{thm::weighted_lb} for the parameters $\eps$ and $h(n)$, and let $\mathcal{G} = \{G\mid G \text{ is the reduction graph of some } \widetilde{G}\in \widetilde{\mathcal{G}}\}$. In addition, recall the construction of $\widetilde{\mathcal{G}}$ and see that there exists a constant $c$ for which, every $\widetilde{G}$ has a minimum vertex cover of size $\Omega(\widetilde{n})$ and is composed of $k = \frac{\widetilde{n}}{4h(\widetilde{n})+2} = O\parenths*{\frac{\widetilde{n}}{h(\widetilde{n})}}$ segments, each contains $4h(\widetilde{n})$ nodes with weight $1$ and $2$ nodes with weight $ch(\widetilde{n})$.
Therefore, the reduction graph $G$ of $\widetilde{G}$ has $n = O\parenths*{\frac{\widetilde{n}}{h(\widetilde{n})}}(4h(\widetilde{n}) + 2ch(\widetilde{n})) = O(\widetilde{n})$ nodes and a minimum vertex cover of size $\Omega(n)$.

\subparagraph*{Existence of a good schedule.}
For the first part, let $G\in \mathcal{G}$ be an $n$-node reduction graph of an $\widetilde{n}$-node graph $\widetilde{G} \in \widetilde{\mathcal{G}}$ with $n\ge 3$. 
Let $v \in \beta \setminus \alpha$ with degree $2$, we
would like to apply \Cref{thm:bicon} to the set $\mathcal{C}$ of all induced subgraphs of $G' := G\setminus\{v\}$, which can be viewed as a bounded hereditary class. Let $H$ be a 2-connected induced subgraph of $G'$ with at least 3 vertices (if there is none, $G'$ admits $(1,1)$-approximation schedules). It follows from the definition of $\mathcal{G}$ that $H$ is a complete bipartite graph $K_{a,b}$, with $1\le a,b\le O(h(\widetilde{n})) = O(h(n))$. For every pair of vertex covers in $K_{a,b}$, there is a  trivial monotone $(1,O(h(n)))$-approximation schedule: add all $\beta$-vertices, then remove all $\alpha$-vertices. \Cref{thm:bicon} then implies that there are similar schedules in $G'$ as well. 

Therefore, $G$ admits a $(1,O(h(n)))$-approximation schedule as well, as given a schedule $\mathcal{S}'$ of $G'$, one can convert it to a $(1,O(h(n))+1)$-approximation schedule of $G$: add $v$ to the cover, then apply $\mathcal{S'}$. As $h(n) = o(n)$ and as $M = \max\{|\alpha|,|\beta|\} = \Omega(n)$, we get that for every large enough $n$, a $(1,O(h(n)))$-approximation schedule implies that the size of all vertex covers in the schedule is bounded by
\begin{align*}
    M + O(h(n)) = M+M\frac{O(h(n))}{M} < M(1+\eps).
\end{align*}
Hence, for large enough $n$, a $(1,O(h(n)))$-approximation schedule can be interpreted as a $(1+\eps,0)$-approximation schedule.

\subparagraph*{Reduction from an unweighted algorithm to a weighted one}
For the second part, let $\mathcal{A}$ be a $t(n)=o(n)$-round algorithm which computes an $(\eta,\xi)$-approximation schedule of length $\ell(n) = O(n)$ on $\mathcal{G}$. We will show that algorithm $\mathcal{A}$ implies an $O(t(\widetilde{n}))$-round algorithm $\widetilde{\mathcal{A}}$ which computes an $(\eta,\xi)$-approximation schedule on $\widetilde{\mathcal{G}}$. Therefore, as $t(\widetilde{n}) = o(h(n))$, by \cref{thm::weighted_lb}, this implies that $\widetilde{A}$ does not compute a $(2-\eps)$-schedule on $\widetilde{\mathcal{G}}$. Combined with \cref{thm::weighted_lb}, we conclude that $\mathcal{A}$ does not compute a $(1+\eps)$-approximation schedule on $\mathcal{G}$ as well.

Algorithm $\widetilde{\mathcal{A}}$, on a $\widetilde{n}$-node graph $\widetilde{G} = (\widetilde{V},\widetilde{E})\in \widetilde{\mathcal{G}}$, simulates $\mathcal{A}$ on the $n$-node reduction graph $G$ of $\widetilde{G}$. The simulation is done by letting each node $\widetilde{v} \in \widetilde{V}$ simulate all nodes in $corr(\widetilde{v}) \subseteq V$. Note that there is no simulation overhead to simulate an algorithm in graph $G$ on communication graph $\widetilde{G}$ as for every $(u,v)\in E$ the edge $(\widetilde{u},\widetilde{v}) \in \widetilde{E}$ exists in the communication network. Therefore, from \cref{lem::weighted_reduction} we conclude that $\Tilde{\mathcal{A}}$ computes an $(\eta,\xi)$-approximation schedule of $G$ with length $\ell(n)$ in $t(n)$ rounds, that is transferred to a schedule of $\widetilde{G}$ with the same length and approximation.
 To interpret $\ell(n)$ and $t(n)$ in terms of $\widetilde{n}$, the number of nodes of $\widetilde{G}$, recall that $n = \frac{\widetilde{n}}{4h(\widetilde{n})+2}(4h(\widetilde{n}) + 2ch(\widetilde{n})) = O(\widetilde{n})$. Therefore, $\widetilde{\mathcal{A}}$ takes $t(n) = t(O(\widetilde{n})) = O(t(\widetilde{n})) = o(\widetilde{n})$ rounds and computes a schedule of length $\ell(n) = \ell(O(\widetilde{n})) = O(\ell(\widetilde{n})) = O(\widetilde{n})$. 
\end{proof}

\subsection{An $\Omega(\logstar n)$ Lower Bound  for Unweighted Cycles}
\label{subsec::unweightedLBcycles}
In the following section, we complement the $O(\logstar(n)/\eps)$ runtime of \cref{cor:variousClassesDist}  with a $\Omega(\logstar n)$ lower bound for computing reconfiguration schedules on unweighted cycles. Our proof uses Ramsey's theorem. 
 Ramsey theory has been used before to prove distributed lower bounds, e.g., for reproving Linial's $\Omega(\logstar n)$ lower bound for coloring cycles with $O(1)$ colors \cite{Suomela2020}, or for lower bounds for approximating independent sets \cite{conf/wdag/CzygrinowHW08}.

Our notations in the following definitions are taken from \cite{radziszowski11ramsey}.
\begin{definition}
Given a set $S$ and an $\ell$-coloring $c:{S\choose k}\rightarrow [\ell]$ of all $k$-subsets of $S$, a subset $S'\subseteq S$ of size $|S'|\ge k$ is \emph{monochromatic} (with color $i$) if the restriction of $c$ to the $k$-subsets of $S'$ is constant (and equals $i$).
\end{definition}

\begin{definition}[Ramsey Number] For  $\ell,k,r\in \mathbb{N}$ with $k\le r$, the \emph{Ramsey Number} $R_{\ell}(k,r)$ is the least $m\in \mathbb{N}$ such that, given a set $S$ of size $m$, for every $\ell$-coloring of $k$-subsets of $S$, there is a monochromatic $r$-subset $S'\subseteq S$.
\end{definition} 

In the lemma below, we use Knuth's arrow notation, where $a\uparrow\uparrow b$ is a tower of $a$'s 
of height $b$. 
\begin{lemma}\label{lem:ramsey}
For all $\ell,k,r\ge 2$, $k\le r$, it holds that $R_\ell(k,r)\le 2\uparrow\uparrow (k\log^*\ell + \log^*(3r\ell))$.
\end{lemma}
\begin{proof}
From \cite[Theorem 1]{10.1112/plms/s3-2.1.417}, we know that $R_{\ell}(k, r) \le \ell * (\ell^{k-1}) * (\ell^{k-2}) * \dots * (\ell^2) * (r\ell)$, where $a*b := a^b$, and $a*b*c=a*(b*c)$. Since the top exponent $2r\ell$ is larger than $k$ and $\ell\ge 2$, it is easy to see that replacing $\ell^i$ with $\ell$ in the expression above, for all $i$, and adding $k$ to the top exponent, does not decrease its value. Thus, we have $R_\ell(k,r)\le \ell*\dots*\ell*(3r\ell)$. Since $\ell\le 2\uparrow\uparrow \log^*\ell$, we replace each $\ell$ with $2\uparrow\uparrow \log^*\ell$ and $3r\ell$ with $2\uparrow\uparrow \log^*(3r\ell)$, to get the claimed bound.
\end{proof}

\begin{theorem}\label{thm:logstarLB}
Let $\eps\in(0,1)$ be a constant, let $\ell : \NN\to\NN$ be a function such that $\ell(n) < n$ and let $\mathcal{C}=\{C_n \mid n\in 2\NN\}$ be the family of even cycles. 
Given an ID space of size $\Omega(n\ell(n)\log{n})$, any $T = T(n,\ell(n))$-round algorithm $\mathcal{A}$ which computes a monotone $\ell(n)$-batch $(2-\eps)$-approximation schedule on $\mathcal{C}$ for every large enough $n$ must have a runtime of $T = \Omega(\log^* (n)/\log^*(\ell(n)))$ rounds.
\end{theorem}

\begin{proof}
Assume, for a contradiction, that $T(n,\ell(n))<\frac{\log^*n}{4\log^*\ell(n)}$, for some sufficiently large even $n$. In the rest of the proof, we let $\ell=\ell(n)$, for simpler notation. Consider $C_n$, and assume, for simplicity, the nodes are labeled with IDs from $[R]$, where $R= \lceil n\ell\log n\rceil$ (the proof can easily adapted for any $R=\Omega(n\ell\log n)$).
Let $\alpha,\beta$, $\alpha\cap\beta=\emptyset$ be  \emph{disjoint minimum vertex covers}; note that the nodes in $C_n$ alternate between $\alpha$ and $\beta$.

\subparagraph*{$\mathcal{A}$ defines a coloring of $(2T+1)$-sized subsets of $R$.} In $T$ rounds, each node $v$ learns the  IDs and type ($\alpha$ or $\beta$) of its $T$-hop neighborhood, i.e., the $2T+1$ nodes closest to it. This information is determined by the $(2T+1)$-tuple of the IDs and the type of $v$, since the types of other nodes are uniquely determined by the type of $v$ and the distance from it. Moreover, we will only focus on  nodes $v$ whose $(2T+1)$-tuple is an increasing sequence; hence, the sequence is determined by the \emph{set} of $(2T+1)$ IDs. The output of $v$ is the batch number where it is processed (recall that the schedule is monotone).  As such, algorithm $\mathcal{A}$ 
corresponds to
a function $\mathcal{A}:\{\alpha,\beta\}\times {R\choose 2T+1} \to \brackets*{\ell}$, which, given the set $S$ of   IDs of the nodes on the $(2T+1)$-path with center node $v$, and the type of $v$, outputs the batch number of $v$.
Our aim is to assign the IDs of nodes so that there is one batch where many $\beta$-nodes are added; therefore, we only focus on $\beta$-nodes, and drop the type parameter of $\mathcal{A}$.

\subparagraph*{There are many monochromatic subsets.}
By the discussion above, we can view  algorithm $\mathcal{A}$ (when restricted to $\beta$-nodes with a monotone ID sequence in the neighborhood) as an $\ell$-coloring of $(2T+1)$-subsets from $[R]$. Let $X=R_{\ell}(2T + 1,2cT + 2T)$, where $c=\lceil 2/\eps\rceil$.  If $R > X$, by the definition of the Ramsey number $R_{\ell}(2T + 1,2cT + 2T)$, there exists a monochromatic subset $S \subseteq [R]$ of size $2cT + 2T$. Hence, one can assign the nodes on an induced path $P \subseteq G$ of length $2cT+2T$ IDs from $S$, in increasing order of IDs. Therefore, as the $T$-hop neighborhood $N^T(v_i)$ of at least $cT$ $\beta$-nodes $v_i\in P$ are all $2T+1$-sized subsets of $S$. As $S$ is monochromatic, each such subset $N^T(v_i)$ is mapped to the same color under $\mathcal{A}$, i.e.,  each of the $cT$ $\beta$-nodes  joins the cover at the same time. 
Now, see that if $R > mX$ 
for some integer $m$, we can take $m$ disjoint subsets of $[R]$, $U_1,\dots,U_{m}$, each of size $X$, and apply the above reasoning for each $U_i$ to obtain $m$ disjoint monochromatic subsets, $S_1,\hdots,S_m$ (there are $m$ separate applications of `Ramsey').
To bound $m$, the number of monochromatic subsets, by \cref{lem:ramsey} we know that $\log^*X\le (2T+1)\log^*\ell + \log^*(12cT\ell)$. Assuming $n$ is large enough with respect to $c$ and $\ell$, and recalling that $T<\frac{\log^* n}{4\log^*\ell}$, we have that $X<\log n$. Therefore, we can take $m := \floor{R/X} \geq \ceil{n\ell\log n}/\log n \geq n\ell$ to obtain disjoint monochromatic subsets $S_1, \dots, S_m$.

\subparagraph*{Many monochromatic subsets have the same color.} While each of the $m$ subsets $S_1,\dots,S_m$ is monochromatic, they might be monochromatic with regard to different colors. But, 
by the pigeonhole principle, there is a color $t\in[\ell]$ such that at least $m/\ell\ge n$ of them are monochromatic with regard to $t$.  Let $S_1,\dots,S_s$ be $s:= \ceil{n/(2cT+2T)}$ of those subsets. We assign the IDs on $C_n$, by partitioning the cycle into paths $P_1,\dots,P_s$, each of length $2cT+2T$ (except perhaps one), and assign the nodes on $P_i$ the IDs from $S_i$, in increasing order of IDs. For each $P_i$, perhaps except for one, it holds that $cT$ of its $\beta$-nodes, namely the ones for which the set of IDs of the $T$-neighborhood is entirely in $S_i$, output $t$. Over the whole graph, at least $(n/(2cT+2T)-1)\cdot cT$ of the $\beta$-nodes output $t$. Just before the $t$-th batch of the schedule,  we have a vertex cover that is at least as large as $M=\max(|\alpha|,|\beta|)=n/2$ (since those are minimum vertex covers). In the $t$-th batch, we add the mentioned $\beta$-vertices, getting a  vertex cover of size
 \[
 M+\frac{ncT}{2cT+2T}-cT=\left(1+\frac{c}{c+1}-\frac{cT}{M}\right)M\ge \left(2-\frac{1}{c+1}-\frac{2c\log^* n}{n}\right)M> (2-\eps)M\ ,
 \]
recalling that $c\ge 2/\eps$, and assuming $n$ is large enough so that $2c\log^*(n)/n<\eps/2$. Thus, algorithm $\mathcal{A}$ gives a schedule whose approximation factor is more than $(2-\eps)$ for the constructed instance. Thus, the assumption that $T<\frac{\log^* n}{4\log^*\ell}$ leads to a contradiction.
\end{proof}

\section*{Acknowledgements}
This project has been supported by the European Union’s Horizon 2020 Research and Innovation Programme under grant agreement no. 755839.

\bibliography{MyBib}

\end{document}